\documentclass[a4paper,10pt]{article}
\usepackage[utf8]{inputenc}
\usepackage{amsthm}
\usepackage{amsmath}
\usepackage{authblk}
\usepackage{mathtools}
\usepackage[mathscr]{euscript}
\usepackage{cite}
\usepackage{amssymb}
\usepackage{nicefrac}
\usepackage{graphicx} 
\usepackage{float}
\usepackage[ruled,vlined]{algorithm2e}
\include{pythonlisting}
\newtheorem{theorem}{Theorem}

\newtheorem{lemma}{Lemma}

\newtheorem{proposition}{Proposition}
\newcommand{\pro}{\text{{\LARGE $\sqcap$}}}
\newcommand{\changes}[1]{#1}

\newtheorem{remark}{Remark}
\title{Analyzing Approximate Value Iteration Algorithms
}
\author{Arunselvan Ramaswamy
\thanks{Heinz Nixdorf Institute and the Department of Computer Science, Paderborn University, 33102 Paderborn, Germany
        \texttt{arunr@mail.upb.de}} \\
 Shalabh Bhatnagar
\thanks{Dept. of Computer Science and Automation and the Robert Bosch Center for Cyber-Physical Systems,
Indian Institute of Science, Bengaluru - 560012, India.
        \texttt{shalabh@iisc.ac.in}}
}
\begin{document}
\maketitle

\begin{abstract}%
In this paper, we consider the stochastic iterative counterpart of the value iteration scheme wherein only noisy and possibly biased approximations of the Bellman operator are available. We call this counterpart as the approximate value iteration (AVI) scheme. Neural networks are often used as function approximators, in order to counter Bellman's curse of dimensionality. In this paper, they are used to approximate the Bellman operator. Since neural networks are typically trained using sample data, errors and biases may be introduced.
The design of AVI accounts for implementations with biased approximations of the Bellman operator and sampling errors. 
We present verifiable sufficient conditions under which AVI is stable (almost surely bounded) and converges to a fixed point of the approximate Bellman operator. To ensure the stability of AVI, we present three different yet related sets of sufficient conditions that are based on the existence of an appropriate Lyapunov function. These Lyapunov function based conditions are easily verifiable and new to the literature. The verifiability is enhanced by the fact that a recipe for the construction of the necessary Lyapunov function is also provided. We also show that the stability analysis of AVI can be readily extended to the general case of set-valued stochastic approximations. Finally, we show that AVI can also be used in more general circumstances, i.e., for finding fixed points of contractive set-valued maps.
\end{abstract}%

\maketitle

%


%
%
%

\section{Introduction}
 \label{sec_intro}
 Reinforcement Learning (RL) is a machine learning paradigm which consists of one or more agents that interact with an environment by taking actions. These actions have consequences for the agent(s) and affect the state of the environment, which is modelled as a Markov decision process (MDP). The typical goal of a reinforcement learning problem is to solve a given sequential decision making problem. Popular RL algorithms include Approximate Value Iteration, Q-Learning, Temporal Difference Learning and Actor-Critic algorithms. Although based on the dynamic programming principle, they differ from traditional dynamic programming algorithms as they do not assume complete knowledge of the MDP. RL algorithms are versatile as they use sample-based methods that work with noisy observations of objective functions such as Value function and Q-factors.
 
Traditional RL algorithms are not readily applicable to decision making problems that involve large state and action spaces, since they suffer from Bellman's curse of dimensionality. This curse is lifted when these algorithms are integrated with artificial neural networks (ANN). Traditional RL algorithms have seen a major resurgence in recent years, as they are used in combination with ANN to solve problems arising in wide ranging applications including transportation, health-care and finance. Such algorithms are popularly referred to as Deep Reinforcement Learning (Deep RL) algorithms. Their popularity was triggered when ``better than'' human performance was demonstrated, in playing Atari and Go, using RL algorithms, \cite{mnih2015human} \cite{silver2017mastering}. The popular roles of ANN within the framework of RL include the approximation of the objective function (Value function, Q-Value Function, etc.) and extraction of relevant features. Deep RL is highly effective owing to recent advances in the field of RL, as well as the benefits resulting from the use of neural network architectures made possible due to significantly enhanced computational capacities. While Deep RL algorithms are empirical miracles, there is very little theory to back them up. In this paper we aim to address this issue.

The main caveat with using neural networks for function approximation is that their architecture needs to be chosen without complete knowledge of the function to be approximated. In fact, there need not exist network weights (hyper-parameters) such that the corresponding network output equals the objective function, everywhere. Further, since the neural networks are trained in an online manner, starting from a random initialization, the transient approximation errors may be large. From the previous argument, it is clear that one cannot expect that these errors would vanish asymptotically. In other words, to understand the long-term behavior of Deep RL algorithms, it is important to consider the effect of these persistent errors. Further, non-diminishing errors have a strong bearing on the stability (almost sure boundedness) of the algorithm at hand. It is well known in the literature that  RL algorithms with unbounded approximation errors are not stable \cite{BertsekasBook}. Hence, the weakest condition under which Deep RL algorithms can be expected to be stable is when the errors are asymptotically bounded. \textit{In this paper, we prove stability under the aforementioned weak condition for Value Iteration algorithms that use the approximate Bellman operator.}

Traditional dynamic programming algorithms require taking expectations, to calculate the value function or the Q-factors. However, modern Deep RL algorithms are applied to scenarios wherein calculating expectations is impossible. Instead, one needs to work with ``data samples''. While the algorithms that can use samples are typically easy to implement, sampling errors often affect stability and convergence properties. In this paper, \textit{we characterize the limiting set of the approximate value iteration algorithm wherein the additive errors that arise from sampling are biased}. \textit{We believe that ours is one of the first results in this direction}.

\subsection{Approximate Value Iteration Algorithm}
Value iteration methods are an important class of RL algorithms that are easy to implement. However, for many important applications, these methods suffer from the \textit{Bellman's curse of dimensionality}. 
In this paper, we present Approximate Value Iteration (AVI) as a means to address the Bellman's curse of dimensionality by
introducing an \textit{approximate Bellman operator} within a classical value iteration framework. We assume that the Bellman operator is approximated in an online manner. Hence, the approximation errors vary over time, albeit in an asymptotically bounded manner.
As stated earlier, if the approximation errors are allowed to be unbounded
 then the algorithm may not converge, see \cite{BertsekasBook} for details. 
 AVIs with bounded approximation errors have been previously studied in
 \cite{BertsekasBook, vroy, munos}. 
 Bertsekas and Tsitsiklis \cite{BertsekasBook} studied scenarios wherein the approximation errors are
 uniformly bounded over all states.
 Munos \cite{munos} extended the analysis of \cite{BertsekasBook}, allowing for approximation errors
 that are bounded in the weighted p-norm sense, for the infinite horizon discounted cost problem.
  In addition to a convergence analysis, \cite{munos} also provides a rate of convergence (finite time) analysis, under the assumption
 that the transition probabilities or future state distributions
 be ``smooth'', among others.

In this paper, we consider the following AVI algorithm:
\begin{equation}
 \label{intro_avi}
 J_{n+1} = J_n + a(n) \left[ TJ_n - J_n + \epsilon_n + M_{n+1} \right],
\end{equation}
where $J_n$ is the current estimate of the optimal cost-to-go vector; $\{a(n)\}_{n \ge 0}$ is the step-size sequence; $T$ is the Bellman operator; $\epsilon_n$ is the approximation error at time $n$; $M_{n+1}$ is a square integrable Martingale difference sequence to account for sampling errors. In Section~\ref{sec_avi} we show that the structure of \eqref{intro_avi} encompasses two important variants that are relevant to the deep RL paradigm. The first variant allows for the use of a (neural network based) function approximation of the Bellman operator, say $AT$, such that the approximation errors are possibly biased. Within the context of deep RL, the neural network is trained in an online manner using a time-varying loss function such as
$
\left[ ATJ_n(s_n) - (r(s_n, a)  + \gamma J_n(s_{n+1})) \right]^2
$ at time $n$. As in Deep Q-Learning, it may be wise to sample mini-batches from an experience replay \cite{mnih2015human}. The second variant allows for the use of sampling, instead of taking expectations. As before, the sampling errors are allowed to be stochastic and biased. For a further discussion on these variants, the reader is referred to Section~\ref{sec_avi}.

An important contribution of this paper is in the weakening of assumptions involved in the analysis of \eqref{intro_avi}. For e.g., we do not require the previously mentioned restriction on transition probabilities of future distributions 
 (cf. Section~\ref{sec_avi}). As a consequence, we only present an asymptotic analysis and not a stronger finite sample analysis.
 However, our analysis encompasses both the 
 stochastic shortest path and the discounted cost infinite horizon problems.
With regards to the stability of \eqref{intro_avi}, it is not immediately clear whether the
 \textit{approximation operator} or the \textit{sampling errors} influence it in a negative way.
 Thus, another important contribution
 of this paper is in proving the stability of AVI under standard assumptions from literature. It should also be noted that we characterise the optimality of the limiting cost-to-go vector found by AVI. Specifically, we show that it is a fixed point of the ``perturbed Bellman operator'' and belongs to a small neighborhood of $J^*$. We further relate the size of this neighborhood to the asymptotic norm-bound on the approximation errors.
   
\subsection{Stochastic Approximation Algorithms}
To develop the sufficient conditions for the convergence of AVI, and for its analysis, we build on tools from the fields of stochastic approximation algorithms (SAs) and viability theory. SAs are an important class of model-free iterative algorithms that are used to develop and analyze algorithms for stochastic control and optimization.
There is a long and rich history to research in the field of SAs, see \cite{robbins} \cite{Benaim96} \cite{BenaimHirsch}  \cite{Borkartt} \cite{Borkar99} \cite{Benaim05}. SAs encompass both the algorithmic and theoretical aspects. Viability theory plays a major role in the latter. For more details on viability theory the reader is referred to \cite{Aubin}. While the main focus of this paper is to understand the long-term behavior of AVI, we deviate from this theme a little and present Lyapunov function based stability 
conditions for set-valued SAs. In other words, \textit{we believe that our stability analysis is readily applicable, verbatim, to  set-valued SAs}.
The analyses presented herein build on the works of \cite{abounadi} and \cite{ramaswamy}. 

\subsection{Fixed point finding algorithms for contractive set-valued maps}
The ideas used to analyze AVI are later used to develop and analyze a SA  for finding fixed points of set-valued maps.
 Fixed point theory is an active area of research due to
 its applications in a multitude of disciplines. 
 Recently, the principle of dynamic programming (DP) was generalized by Bertsekas \cite{AbstractDP} to solve problems which, previously, could not be solved using classical DP. This extension involved a new abstract definition of the Bellman operator. The theory thus developed is called Abstract Dynamic Programming. An integral component of this new theory involves showing that the solution to the abstract Bellman operator is its fixed point. We believe that the results of Section \ref{sec_fp} are helpful in solving problems that can be formulated as an  Abstract Dynamic Program.
 Our contribution on this front is in the development and analysis of a
 \textit{SA for finding fixed points of contractive 
 set-valued maps,
 } see Section~\ref{sec_fp} for details. As
 mentioned before, we show that such algorithms are bounded almost surely and that they converge to a 
 sample path dependent fixed point of the set-valued map under consideration.
 \textit{To the best of our knowledge
ours is the first SA, complete with analysis, for finding fixed points of set-valued maps}.
\section{Definitions and Notations} \label{sec_def}
Key definitions and notations encountered in this paper are listed in this section.
{\renewcommand\labelitemi{}
\begin{itemize}
 \item \textbf{\textbf{[Upper-semicontinuous map]}} We say that $H$ is upper-semicontinuous,
  if given sequences $\{ x_{n} \}_{n \ge 1}$ (in $\mathbb{R}^{d_1}$) and 
  $\{ y_{n} \}_{n \ge 1}$ (in $\mathbb{R}^{d_2}$)  with
  $x_{n} \to x$, $y_{n} \to y$ and $y_{n} \in H(x_{n})$, $n \ge 1$, 
  then $y \in H(x)$.
\item
\textbf{\textbf{[Marchaud Map]}} A set-valued map $H: \mathbb{R}^{d_1} \to \{subsets\ of\ \mathbb{R}^{d_2}$\} 
is called \textit{Marchaud} if it satisfies
the following properties:
 \textbf{(i)} for each $x$ $\in \mathbb{R}^{d_1}$, $H(x)$ is convex and compact;
 \textbf{(ii)} \textit{(point-wise boundedness)} for each $x \in \mathbb{R}^{d_1}$,  
 $\underset{w \in H(x)}{\sup}$ $\lVert w \rVert$
 $< K \left( 1 + \lVert x \rVert \right)$ for some $K > 0$;
 \textbf{(iii)} $H$ is \textit{upper-semicontinuous}. \\
Let $H$ be a Marchaud map on $\mathbb{R}^d$.
The differential inclusion (DI) given by
\begin{equation} \label{di}
\dot{x} \ \in \ H(x),
\end{equation}
is then guaranteed to have at least one solution that is absolutely continuous. 
The reader is referred to \cite{Aubin} for more details.
We say that $\textbf{x} \in \sum$ if $\textbf{x}$ 
is an absolutely continuous map that satisfies (\ref{di}).
The \textit{set-valued semiflow}
$\Phi$ associated with (\ref{di}) is defined on $[0, + \infty) \times \mathbb{R}^d$ as: \\
$\Phi_t(x) = \{\textbf{x}(t) \ | \ \textbf{x} \in \sum , \textbf{x}(0) = x \}$. Let
$B \times M \subset [0, + \infty) \times \mathbb{R}^d$ and define
\begin{equation}\nonumber
 \Phi_B(M) = \underset{t\in B,\ x \in M}{\bigcup} \Phi_t (x).
\end{equation}
\item \textbf{\textbf{[Limit set of a solution \& $\omega$-limit-set]}} The limit set of a solution $\textbf{x}$
with $\textbf{x}(0) = x$ is given by
$L(x) = \bigcap_{t \ge 0} \ \overline{\textbf{x}([t, +\infty))}$.
Let $M \subseteq \mathbb{R}^d$, the $\omega$-limit-set be defined by
$
 \omega_{\Phi}(M) = \bigcap_{t \ge 0} \ \overline{\Phi_{[t, +\infty)}(M)}.
$
\item \textbf{\textbf{[Invariant set]}}
$M \subseteq \mathbb{R}^d$ is \textit{invariant} if for every $x \in M$ there exists 
a trajectory, $\textbf{x} \in \sum$, entirely in $M$
with $\textbf{x}(0) = x$, $\textbf{x}(t) \in M$,
for all $t \ge 0$. Note that the definition of invariant set used in this paper, is the same as
that of positive invariant set used in \cite{Benaim05} and \cite{BorkarBook}.
\item \textbf{\textbf{[Open and closed neighborhoods of a set]}}
Let $x \in \mathbb{R}^d$ and $A \subseteq \mathbb{R}^d$, then
$d(x, A) : = \inf \{\lVert x- y \rVert \ | \ y \in A\}$. We define the $\delta$-\textit{open neighborhood}
of $A$ by $N^\delta (A) := \{x \ |\ d(x,A) < \delta \}$. The 
$\delta$-\textit{closed neighborhood} of $A$ 
is defined by $\overline{N^\delta} (A) := \{x \ |\ d(x,A) \le \delta \}$.
The open ball of radius $r$ around the origin is represented by $B_r(0)$,
while the closed ball is represented by $\overline{B}_r(0)$.
\item \textbf{\textbf{[Internally chain transitive set]}}
$M \subset \mathbb{R}^{d}$ is said to be
internally chain transitive if $M$ is compact invariant and for every $x, y \in M$,
$\epsilon >0$ and $T > 0$ we have the following: There exists $n$ and $\Phi^{1}, \ldots, \Phi^{n}$ that
are $n$ solutions to the differential inclusion $\dot{x}(t) \in H(x(t))$,
points $x_1(=x), \ldots, x_{n+1} (=y) \in M$
and $n$ real numbers 
$t_{1}, t_{2}, \ldots, t_{n}$ greater than $T$ such that: $\Phi^i_{t_{i}}(x_i) \in N^\epsilon(x_{i+1})$ and
$\Phi^{i}_{[0, t_{i}]}(x_i) \subset M$ for $1 \le i \le n$. The sequence $(x_{1}(=x), \ldots, x_{n+1}(=y))$
is called an $(\epsilon, T)$ chain in $M$ from $x$ to $y$.
\item \textbf{\textbf{[Attracting set \& fundamental neighborhood]}}
$A \subseteq \mathbb{R}^d$ is \textit{attracting} if it is compact
and there exists a neighborhood $U$ such that for any $\epsilon > 0$,
$\exists \ T(\epsilon) \ge 0$ with $\Phi_{[T(\epsilon), +\infty)}(U) \subset
N^{\epsilon}(A)$. Such a $U$ is called the \textit{fundamental neighborhood} of $A$. 
In addition to being compact if the \textit{attracting set} is also invariant then
it is called an \textit{attractor}.
The \textit{basin
of attraction } of $A$ is given by $B(A) = \{x \ | \ \omega_\Phi(x) \subset A\}$.
\item \textbf{[Global attractor]} If the basin of a given attractor is $\mathbb{R}^d$, then the attractor is called global attractor.
\item \textbf{[Globally asymptotically stable equilibrium point]} A point $x_0$ is an equilibrium point for the DI \eqref{di}, if $0 \in H(x_0)$. Further, it is globally asymptotically stable if it is a global attractor. This notion is readily extensible to sets.
\item \textbf{\textbf{[Lyapunov stable]}} The above set $A$ is Lyapunov stable 
if for all $\delta > 0$, $\exists \ \epsilon > 0$ such that
$\Phi_{[0, +\infty)}(N^\epsilon(A)) \subseteq N^\delta(A)$.
\end{itemize}}
\section{Approximate value iteration methods} \label{sec_avi}
Most value iteration methods are based on fixed point finding algorithms. This is because the optimal cost-to-go vector is a fixed point of the Bellman operator.
Since the Bellman operator is contractive with respect to some weighted max-norm $\lVert \cdotp \rVert_\nu$, it follows from fixed point theory that there is a unique fixed point for it. Further, this fixed point is the required optimal cost-to-go vector. 
Suppose $T$ is the Bellman operator, the aim of value iteration methods is to find $J^*$ such that $J^* = TJ^*$. Given a cost-to-go vector $J = (J(s), \ s \in \mathcal{S})^{\footnotesize T}$, let
\begin{equation} \label{avi_bellman}
TJ(s) = \underset{a \in \mathcal{A}}{\max}\ \mathbb{E}_{s' \sim \mathcal{P}} \left[ r(s,a) + \gamma J(s') \mid s, a \right],
\end{equation}
where $\mathcal{A}$ is the action space and $\mathcal{S}$ is the state space with $s, s' \in \mathcal{S}$; $r(s, a)$ is the single-stage cost when action $a$ is chosen in state $s$; $\mathcal{P}$ is the (unknown) transition probability law with $\mathcal{P}(s, a, s')$ being the probability of transition from state $s$ to $s'$ when action $a$ is taken and $0 < \gamma \le 1$ is the discount factor. When $0 < \gamma < 1$ we are in the infinite horizon discounted cost problem setting, and $\gamma = 1$ corresponds to the setting of the infinite horizon stochastic shortest path problem. In this paper we do not distinguish between the two, we shall implicitly work with the appropriate definition of $T$. 

Suppose $T$ can be exactly calculated, then the recursion $J_{n+1} \leftarrow T J_n$ converges to $J^*$ starting from any $J_0$. However, for an exact calculation of $T$, one requires complete knowledge of the transition probability law $\mathcal{P}$ and the reward function $r(\cdotp, \cdotp)$. In many applications this could be a hard requirement to satisfy. In today's deep learning age, \textit{it is common to work with approximations of the Bellman operator}. These approximations are noisy (stochastic) and biased. Below, we present \textbf{Approximate Value Iteration (AVI)}, a stochastic iterative counterpart of traditional value iteration, designed to operate in the presence of noise and approximations:

\begin{equation} \label{avi_algo}
J_{n+1} = J_n + a(n) \left[ ATJ_n - J_n  + M_{n+1} \right],
\end{equation}
where $J_n \in \mathbb{R}^d$ for all $n \ge 0$; $A$ is the approximation operator; $T$ is the Bellman operator, see \eqref{avi_bellman}; $\{a(n)\}_{n \ge 0}$ is the given step-size sequence and $\{M_{n+1}\}_{n \ge 0}$ is the noise sequence. Note that $J_n = (J_n(1), \ldots, J_n(d))^{\footnotesize T}$ for $n \ge 0$ and the state space is given by $\mathcal{S} = \{1, \ldots, d \}$. Let us rewrite \eqref{avi_algo} as the following:

\begin{equation}
 \label{avi_avi}
 J_{n+1} = J_n + a(n) \left[ TJ_n - J_n + \epsilon_n + M_{n+1} \right],
\end{equation}
where $\epsilon_n := ATJ_n - TJ_n$ is the approximation error at stage $n$; $\{M_{n}\}_{n \ge 1}$ is a square integrable Martingale difference noise sequence that is adapted to the filtration $\{ \mathcal{F}_n \}$, defined by $\mathcal{F}_n := \sigma \langle J_m, \epsilon_m \mid m \le n \rangle$, $n \ge 0$. 
\subsection{Asssumptions to analyze AVI} \label{avi_asmp}
Before listing the assumptions required to analyze \eqref{avi_algo}/\eqref{avi_avi}, let us define the weighted max-norm $\lVert \cdotp \rVert_\nu$. Given $\nu = (\nu_1, \ldots, \nu_d)^{\footnotesize T}$ such that
$\nu_i > 0$ for $1 \le i \le d$, let $\lVert x \rVert_\nu = \max \left\{ \frac{|x_i|}{\nu_i}  \mid
1 \le i \le d \right\}$, where $x = (x_1, x_2, \ldots, x_d)^{\footnotesize T} \in \mathbb{R}^d$.
\begin{itemize}
 \item[\textbf{(AV1)}] The Bellman operator $T$ is contractive with respect to some weighted max-norm, 
 $\lVert \cdotp \rVert_\nu$, \textit{i.e.,} $\lVert Tx - Ty \rVert_\nu \le \alpha \lVert x - y \rVert_\nu$
 for some $0 < \alpha < 1$.
 \item[\textbf{(AV2)}] $T$ has a unique fixed point $J^*$ and $J^*$ is the 
 unique globally asymptotically stable equilibrium point of $\dot{J}(t) = TJ(t) - J(t)$.
 \item[\textbf{(AV3)}] Almost surely $\limsup \limits_{n \to \infty}\ \lVert \epsilon_n \rVert_\nu \le \epsilon$, for some fixed $\epsilon > 0$.
 \item[\textbf{(AV4)}] The step-size sequence $\{a(n)\}_{n \ge 0}$ is such that $\forall n, \ a(n) \ge 0$, $\sum_{n \ge 0} a(n) = \infty$ and $\sum_{n \ge 0} a(n)^2 < \infty$.
 \item[\textbf{(AV5)}] $(M_n, \mathcal{F}_n)_{n \ge 1}$ is a square integrable Martingale difference sequence $\bigg( E[M_{n+1} \mid \mathcal{F}_n] = 0$ and $EM_{n+1}^2 < \infty$, $n \ge 0 \bigg)$ such that
$
 E\left[ \lVert M_{n+1} \rVert ^2 \mid \mathcal{F}_n \right] \le K(1 + \lVert J_n \rVert ^2), \text{ where $n \ge 0$ and $K > 0$.}
$. The filtration $\{ \mathcal{F}_n\}_{n \ge 0}$ is as defined above.
\end{itemize}

Now, we briefly discuss the above listed requirements. Assumption $(AV1)$ is standard in literature and is readily satisfied in many applications, see \textit{Section 2.2}
of Bertsekas and Tsitsiklis \cite{BertsekasBook} for details. In Section \ref{sec_ana_avi}, we discuss how (AV2) ensures the stability of AVI. $(AV3)$ requires that the \textbf{stochastic approximation errors are asymptotically bounded in an almost sure sense}. This asymptotic bound is with respect to the weighted max-norm. Later, in Section~\ref{sec_avi1}, we show that the analysis of (\ref{avi_avi}) is unaltered
when the approximation errors are more generally bounded in the weighted p-norm sense 
(weighted Euclidean norms). Let us say that $(AV3)$ is violated. This implies that $\lVert \epsilon_{n(m)} \rVert \uparrow \infty$ along a sequence $\{n(m)\}_{m \ge 0} \subseteq \mathbb{N}$. In words, there is a massive failure in approximating the Bellman operator, and there are points wherein the approximate Bellman operator differs from the true one by large amounts. Therefore, $\sum \limits_{n \ge 0} a(n) \lVert \epsilon_n \rVert$ may equal $\infty$ as a consequence, and we can never guarantee stability. Assumption $(AV3)$ essentially requires that such circumstances be avoided.

On the surface $(AV5)$ concerns the additive noise terms that are modelled as a square integrable martingale difference sequence. However, it also serves the dual role of analyzing the \textbf{``sample-based'' variant of AVI}. To illustrate this, we assume that the single-stage reward function $r(\cdotp, \cdotp)$ is a given deterministic function, and that the algorithm can sample from the transition probability law $\mathcal{P}(s,a,\cdotp)$. If we use the notation $\hat{T}$ to represent sample-Bellman operator, then
\[
\hat{T} J(s) := \max \limits_{a \in \mathcal{A}} \left[ r(s,a) + J(\psi(s,a))\right],
\]
where $\psi(s,a)$ is a random variable that takes values in the state-space $\mathcal{S}$, and is distributed according to $\mathcal{P}(s,a,\cdotp)$. This is the setting of deep learning, and the sample-based variant of AVI, is given by
\begin{equation}
\label{avi_avi_sample}
J_{n+1} = J_n + a(n) \left[ \hat{T} J_n - J_n + \epsilon_n \right].
\end{equation}
If we condition on the current state and action, it is fair to assume that the $\psi(s,a)$ samples required to calculate the RHS of \eqref{avi_avi_sample}, are independent across time. Specifically, if we codify all the samples taken at stage $n$ as $\psi_n$, it follows that $\{\psi_n\}_{n \ge 0}$ is an independent sequence. First, we define the filtration as: $\mathcal{F}_0 = \sigma \langle J_0, \epsilon_0 \rangle$ and $\mathcal{F}_n = \sigma \langle J_m, \epsilon_m, \psi_k \mid m \le n, \ k < n \rangle$, $n \ge 1$. Next, we define a zero-mean square integrable martingale difference sequence as: $M_{n+1} := \left( \hat{T} J_n - J_n - \mathbb{E} \left[ \hat{T} J_n - J_n \mid \mathcal{F}_n \right] \right)$, $n \ge 0$. It follows from the definition of the filtration, that $\mathbb{E} \left[ \hat{T} J_n - J_n \mid \mathcal{F}_n \right] = TJ_n - J_n$ for all $n \ge 0$, where $T$ is the Bellman operator as defined in \eqref{avi_bellman}. Hence, the sample-based AVI \eqref{avi_avi_sample} can be written as \eqref{avi_avi}, with $M_{n+1}$ as above. In other words, \textbf{the AVI algorithm given by \eqref{avi_avi} has a general architecture, and can be used to solve problems involving biased function approximations and noisy samples}.

Since $M_{n+1} = \left( \hat{T} J_n - J_n \right) - \left( TJ_n - J_n \right)$, its component along the dimension associated with state $s$ is given by $\hat{T} J_n(s) - T J_n(s)$. It follows from the definitions of $\hat{T}$ and $T$ that:
\[
\left| \hat{T} J_n(s) - T J_n(s) \right| \le \underset{a \in \mathcal{A}, s' \in \mathcal{S}}{\max} |r(s',a)| + \underset{a \in \mathcal{A}}{\max} \left( \left| J_n(\psi(s,a)) \right| + \mathbb{E}_{s' \sim \mathcal{P}} \left| J_n(s') \right| \right).
\]
Let us use $\lVert \cdotp \rVert_{\infty}$ to represent the max-norm. There exists $K_1 > 0$ such that the above inequality becomes:
\[
\left| \hat{T} J_n(s) - T J_n(s) \right| \le K_1 (1 + \lVert J_n \rVert_\infty).
\]
For non-negative $a, b \in \mathbb{R}$, we have that $(a + b)^2 \le 2(a^2 + b^2)$ (using AM-GM inequality). Hence
\[
\left| \hat{T} J_n(s) - T J_n(s) \right|^2 \le 2K_1^2 (1 + \lVert J_n \rVert_\infty ^2).
\]
As the max-norm is bounded above by the Euclidian norm ($\lVert \cdotp \rVert_\infty ^2 \le \lVert \cdotp \rVert ^2$), we conclude from the above inequality that:
\[
E\left[ \lVert M_{n+1} \rVert ^2 \mid \mathcal{F}_n \right] \le K(1 + \lVert J_n \rVert ^2), \text{ where $K := 2d K_1^2$.}
\]
In other words, the \textbf{sample-based AVI satisfies $(AV5)$}. If we are able to show that $\underset{n \ge 0}{\sup} \ \lVert J_n \rVert < \infty$ a.s., then it follows from $(AV4)$ that the quadratic variation process associated with the Martingale sequence is a.s. bounded. It then follows from the Martingale Convergence Theorem that $\sum \limits_{m=0}^n a(m) M_m$, $n \ge 0$, converges almost surely. If we interpret the martingale difference term $M_{n+1}$ as the sampling error at stage $n$, then we may conclude that the \textbf{sampling errors vanish over time}. 

As stated in the Abstract and the Introduction, we are interested in the stability and convergence analysis of AVI. Before we present this analysis, we take a detour and consider general set-valued SAs, in order to \textbf{present the required (Lyapunov function based) stability assumptions}. Clearly the AVI given by \eqref{avi_avi} can be viewed as appropriate set-valued SA. We present the stability assumptions for the general setting of set-valued SAs, since they can be additionally applied to analyze other algorithms arising in Deep RL and stochastic approximation settings. Following the short detour, we shall return to the analysis of AVI.  
\section{Lyapunov stability assumptions and general set-valued stochastic approximations} \label{sec_sa}
Before presenting the Lyapunov stability conditions, we present the general structure of set-valued SAs. Consider the following iteration in $\mathbb{R}^d$:
\begin{equation} \label{sa_sri}
x_{n+1} = x_n + a(n) \left[ y_n + M_{n+1} \right],
\end{equation}
where $y_n \in H(x_n)$ for all $n \ge 0$ with $H: \mathbb{R}^d \to \{\text{subsets of }\mathbb{R}^d\}$,
$\{a(n)\}_{n \ge 0}$ is the given step-size sequence and $\{ M_{n+1}\}_{n \ge 0}$ is the given noise sequence.
In addition to the required stability assumptions, we list a few that are required to present the stability analysis of \eqref{sa_sri} \footnote{
$(A3a)$ is the required Lyapunov-based stability condition. Two
additional alternative stability conditions, \textit{viz.,} $(A3b)$ and $(A3c)$ will also be presented.
}.
\begin{itemize}
\item[\textbf{(A1)}] $H: \mathbb{R}^d \to \{\text{subsets of }\mathbb{R}^d\}$ is a Marchaud map.
\item[\textbf{(A2)}] \begin{itemize}
			\item[(i)] For all $n \ge 0$, $\lVert M_{n+1} \rVert \le D$, where $D >0$ is some constant.
			\item[(ii)] $\underset{k \to \infty}{\lim}\sum\limits_{n=k}^{m_T(k)} a(n) 
			M_{n+1} =  0$ for all $T$, where $m_T(k) := \min \left\{ m \ge k \mid \sum_{n=k}^m a(n) \ge T \right\}$.
			\end{itemize}
\end{itemize}
Note that the assumption on the Martingale noise terms, $(A2)$, is stronger than the corresponding one in Section~\ref{sec_avi}, i.e., $(AV5)$. We consider this stronger assumption for the sake of clarity in presenting the stability analysis of \eqref{sa_sri}. Later, in Section~\ref{sec_gn}, we show that the aforementioned stability analysis carries forward even under the more general $(AV5)$ instead of $(A2)$. Below, we present a key assumption required to prove the stability of \eqref{avi_avi} and \eqref{sa_sri}. Additionally, we present two different yet related variants of this key assumption, both of which lead to identical analyses. The verification of these assumptions involves the construction of an associated Lyapunov function. A recipe for its construction is discussed in Remark~\ref{lyapunov}. These Lyapunov stability assumptions provide an alternative to the ones in \cite{ramaswamy}. Further, in lieu of Remark~\ref{lyapunov}, we believe that these assumptions are readily verifiable. 

\subsection*{Lyapunov function based stability assumptions}
We begin by recalling the definition of the set-valued semiflow $\Phi$ from Section~\ref{sec_def}. Given a solution $\textbf{x}$ of the DI $\dot{x}(t) \in H(x(t))$, 
$\Phi_t(x) := \{\textbf{x}(t) \ | \ \textbf{x} \in \sum , \textbf{x}(0) = x \}$, and $\sum$ is the set of solutions.
\begin{itemize}
\item[\textbf{(A3a)}] Associated with the differential inclusion (DI) $\dot{x}(t) \in H(x(t))$ 
is a compact set $\Lambda$, a bounded open neighborhood $\mathcal{U}$ 
$\left( \Lambda \subseteq \mathcal{U}\subseteq \mathbb{R}^d \right)$ and a
function $V: \overline{\mathcal{U}} \to \mathbb{R}^+$ such that
	\begin{itemize}
	\item[$(i)$] $\forall t \ge 0$, $\Phi_t (\mathcal{U}) \subseteq \mathcal{U}$ \textit{i.e.,} $\mathcal{U}$ is strongly positively invariant.
	\item[$(ii)$] $V^{-1} (0) = \Lambda$.
	\item[$(iii)$] $V$ is a continuous function such that for all
	$x \in \mathcal{U} \setminus \Lambda$ and $y \in \Phi_t (x)$ we have $V(x) > V(y)$, for any $t > 0$.
	\end{itemize}
\end{itemize}

It follows from $(A3a)$ and Proposition 3.25 of Bena\"{i}m, Hofbauer and Sorin \cite{Benaim05} that $\Lambda$ contains a Lyapunov stable attracting set. Further there exists an attractor contained in $\Lambda$ whose basin
of attraction contains $\mathcal{U}$. 
Let us define $\mathcal{V}_s := \{x \mid V(x) < s\}$ and $\mathcal{V}_{\overline{r}} := \{x \mid V(x) \le r\}$, for every $s > 0$ and $r \ge 0$. Then, using the compactness of $\Lambda$ and
$(A3a)(ii)$ we can show that $\underset{r > 0}{\bigcap} \mathcal{V}_r = \Lambda$. Since $\mathcal{U}$ is open, $\exists \  0< R(a) < \sup \limits_{x \in \mathcal{U}} V(x)$, such that $\mathcal{V}_r \subset \mathcal{U}$, for $r \le R(a)$. Now, we define open sets $\mathcal{B}$ and $\mathcal{C}$ such that the following conditions are satisfied: 
\begin{itemize}
\item[(Ca)] $\Lambda \subseteq \overline{\mathcal{B}} \subseteq \mathcal{C}$
\item[(Cb)] $\overline{\mathcal{C}} \subseteq \mathcal{U}$
\end{itemize}
We use the notations $\mathcal{B}_a$ and $\mathcal{C}_a$ to indicate that the constructed sets are associated with (A3a).
Let $\mathcal{B}_a := \mathcal{V}_{r(a)}$ and $\mathcal{C}_a : = \mathcal{V}_{s(a)}$ such that $0<r(a) < s(a) < R(a)$, where $R(a)$ is such that $\mathcal{V}_q \subset \mathcal{U}$ for $q \le R(a)$. Since $V$ is continuous, we have that the closures of $\mathcal{V}_{r(a)}$ and $\mathcal{V}_{s(a)}$ satisfy $\overline{\mathcal{V}}_{r(a)} = \mathcal{V}_{\overline{r(a)}}$ and $\overline{\mathcal{V}}_{s(a)} = \mathcal{V}_{\overline{s(a)}}$, respectively. Hence, conditions (Ca) and (Cb) are readily satisfied.
\paragraph{}
The purpose of defining such sets will be clarified shortly. For now, we proceed with the first variant of $(A3a)$, we call it $(A3b)$:
\begin{itemize}
\item[\textbf{(A3b)}] 
Associated with $\dot{x}(t) \in H(x(t))$ is a compact set $\Lambda$, a bounded open neighborhood $\mathcal{U}$ and a
function $V: \overline{\mathcal{U}} \to \mathbb{R}^+$ such that
	\begin{itemize}
	\item[(i)] $\forall t \ge 0$, $\Phi_t (\mathcal{U}) \subseteq \mathcal{U}$ \textit{i.e.,} $\mathcal{U}$ is strongly positively invariant.
	\item[(ii)] $V^{-1} (0) = \Lambda$.
	\item[(iii)] $V$ is an upper semicontinuous function such that for all $x \in \overline{\mathcal{U}} \setminus \Lambda$ and $y \in \Phi_t (x)$ we have $V(x) > V(y)$, where $t > 0$.
	\item[(iv)] $\mathcal{V}_{\overline{r}} := \{x \mid V(x) \le r\}$ is closed for each $r \ge 0$. 
	\end{itemize}
\end{itemize}

	Since $(A3a)(iii) \implies (A3b)(iii)$, one may view $(A3b)$ as a weakening of $(A3a)$. Again, using \textit{Proposition 3.25} of Bena\"{i}m, Hofbauer and Sorin \cite{Benaim05} we get that  
	$\Lambda$ contains an attractor set such that $\mathcal{U}$ belongs to its basin of attraction. As in the case of
	$(A3a)$, we wish to define open sets $\mathcal{B}_b$ and $\mathcal{C}_b$ satisfying the previously mentioned conditions (Ca) and (Cb). We begin by claiming that $\mathcal{V}_r$ is open for $r > 0$. We prove this claim by showing that $\mathcal{V}_r^c = \{x \mid V(x) \ge r\}$ is closed. For this, we consider $\{x_n\}_{n \ge 0} \subseteq \mathcal{V}_r ^c$ such that $\lim \limits_{n \to \infty} x_n = x$. From the upper semicontinuity of $V$, we get $V(x) \ge \limsup \limits_{n \to \infty} V(x_n) \ge r$, hence, $x \in \mathcal{V}_r ^c$. Thus we get that $\mathcal{V}_r$ is open and $\mathcal{V}_{\overline{r}}$ is closed (from (A3b)(iv)) for $r > 0$. Finally, note that $\Lambda = \underset{r > 0}{\bigcap} \mathcal{V}_r$ as a consequence of (A3b)(ii). Hence, as in the case of (A3a), $\exists \ 0 < R(b) < \sup \limits_{x \in \mathcal{U}} V(x)$ such that $\mathcal{V}_r \subset \mathcal{U}$ for $r \le R(b)$.
	
	We are now ready to define $\mathcal{B}_b$ and $\mathcal{C}_b$ satisfying conditions (Ca) and (Cb). As before, we let $\mathcal{B}_b := \mathcal{V}_{r(b)}$ and $\mathcal{C}_b := \mathcal{V}_{s(b)}$, where $0 < r(b) < s(b) < R(b)$. From (A3b)(iv), we get that $\overline{\mathcal{V}}_{r(b)} \subset \mathcal{V}_{\overline{r(b)}}$ and $\overline{\mathcal{V}}_{s(b)} \subset \mathcal{V}_{\overline{s(b)}}$. Using this observation, we can easily conclude that conditions (Ca) and (Cb) are satisfied.  

	\vspace*{.3cm}
	
	\begin{remark}
	\label{lyapunov}
	If we can associate $\dot{x}(t) \in H(x(t))$ with an attractor 
	set $\mathcal{A}$ and a strongly positive invariant neighborhood 
	$\mathcal{U}$, of $\mathcal{A}$, then we can define an upper-semicontinuous Lyapunov function
	$V$, as found in \textit{Remark 3.26, Section 3.8} of \cite{Benaim05}. In particular, we can define a local Lyapunov function
	$V: \overline{\mathcal{U}} \to \mathbb{R}^+$ such that $V(x) :=$ $\max 
	\left\{ d(y, \mathcal{A}) g(t) \mid y \in \Phi_t(x), t \ge 0 \right\}$,
	where $g$ is an increasing function with $0 < \overline{c} < g(t) < \overline{d}$ for all $t \ge 0$.
	
	We claim that $V$, as defined above, satisfies $(A3b)$. To see this, we begin by noting that $(A3b)(i)$ is trivially satisfied. If we let $\Lambda := \mathcal{A}$, then it follows from the definition of $V$ that $(A3b)(ii)$ is satisfied.
	Since $\mathcal{U}$ is strongly positive invariant and
	$V(x) \le \underset{u \in \mathcal{U}}{\sup} \ d(u, \mathcal{A}) \times \overline{d}$ 
	for $x \in \mathcal{U}$,
	$\underset{u \in \mathcal{U}}{\sup} \ V(u) < \infty$. It now follows from the upper semicontinuity of $V$ that $\underset{u \in \overline{\mathcal{U}}}{\sup} \ V(u) < \infty$ \textit{i.e.,} (A3b)(iv) is satisfied.
	To show that $(A3b)(iii)$ is also satisfied, we first fix $x \in \mathcal{U}$ and $t > 0$. It follows from the definition of a semi-flow that
	$\Phi_s(y) \subseteq \Phi_{t+s}(x)$ for any $y \in \Phi_t (x)$, where $s > 0$. Further, 
	\[
	V(x) \ge \max \{d(z,  \mathcal{A}) g(t+s) \mid z \in \Phi_s(y), s \ge 0\} \text{ and}
	\]
	\[
	\max \{d(z,  \mathcal{A}) g(t+s) \mid z \in \Phi_s(y), s \ge 0\} > \max \{d(z,  \mathcal{A}) g(s) \mid z \in \Phi_s(y), s \ge 0\}.
	\]
	The RHS of the above equation is $V(y)$ \textit{i.e.,} $V(x) > V(y)$.
	
	It is left to show that $\mathcal{V}_{\overline{r}}$ is closed for $r \ge 0$. For this, we present a proof by contradiction, by assuming that $\{x_n\}_{n \ge 0} \subset \mathcal{V}_{\overline{r}}$, $\lim \limits_{n \to \infty} x_n = x$ and $x \notin \mathcal{V}_{\overline{r}}$. We have $V(x) = r + c$ for some $c > 0$. From the definition of $V$, we get that $\exists \ t(x) \ge 0$ and $y(x) \in \Phi_{t(x)}(x)$ such that $d(y(x), \mathcal{A}) g(t(x)) > r + c/2$. We may use \textit{Corollary 2 (Approximate Selection Theorem)} from \textit{Chapter2, Section 2} of \cite{Aubin} to construct $\{y_n\}_{n \ge 0}$ such that $y_n \in \Phi_{t(x)}(x_n)$ for each $n \ge 0$ and $\liminf \limits_{n \to \infty} y_n = y(x)$. Hence $\exists \ N$ such that $d(y_N, \mathcal{A}) g(t(x)) > r + c/4$. This yields a contradiction, as $x_N \in \mathcal{V}_{\overline{r}}$ implying that $d(y_N, \mathcal{A}) g(t(x)) \le r$.
	\end{remark} 
	
	\vspace*{.5cm}
	
Below is the final variant of $(A3a)$:
\begin{itemize}
\item[\textbf{(A3c)}]
\begin{itemize}
 \item[(i)] $\mathcal{A}$ is the global attractor of $\dot{x}(t) \in H(x(t))$.
 \item[(ii)] $V: \mathbb{R}^d \to \mathbb{R}^+$ is an
upper semicontinuous function such that $V(x) > V(y)$ for all $x \in \mathbb{R}^d \setminus \mathcal{A}$, 
$y \in \Phi_t (x)$ and $t > 0$.
\item[(iii)] $V(x) \ge V(y)$ for all $x \in \mathcal{A}$, $y \in \Phi_t (x)$ and $t > 0$.
\item[(iv)] $\mathcal{V}_{\overline{r}} := \{x \mid V(x) \le r\}$ is closed for each $r \ge 0$.
\end{itemize}
\end{itemize}

Since $\mathcal{A}$ is the global attractor, it is compact and for every $x \notin \mathcal{A}$, we have that $V(x) \ge \sup \limits_{y \in \mathcal{A}} V(y)$ and $\mathcal{A} = \underset{r > \sup \limits_{y \in \mathcal{A}} V(y)}{\bigcap} \mathcal{V}_r$. Further, $\exists \ R(c)$ with $\sup \limits_{y \in \mathcal{A}} V(y) <  R(c) < \infty$ such that $\mathcal{V}_r$ is a bounded open set for $r \le R(c)$.
Again, we define open sets $\mathcal{B}_c$ and $\mathcal{C}_c$ satisfying 
conditions (Ca) and (Cb) (see below the statement of $(A3a)$), as $\mathcal{B}_c := \mathcal{V}_{r(c)}$ and $\mathcal{C}_c := \mathcal{V}_{s(c)}$, $0 < r(c) < s(c) < R(c)$. The steps involved in showing that these are indeed the required sets follow similar arguments as for the sets associated with (A3b).

$\\$
\begin{remark}
 \label{lyapunov1}
 In Remark~\ref{lyapunov}, we explicitly constructed a local Lyapunov function satisfying
 $(A3b)$. We can similarly construct a global Lyapunov function, $\hat{V}: \mathbb{R}^d \to \mathbb{R}^+$,  satisfying $(A3c)$ as follows: $\hat{V}(x) := \ \max\left\{ d(y, \mathcal{A}) g(t) \mid y \in \Phi_t(x), t \ge 0 \right\}$, with $g(\cdotp)$ defined as in Remark~\ref{lyapunov}. 
 To see that $\hat{V}$ satisfies $(A3c)$, one may emulate the arguments following Remark~\ref{lyapunov}.
\end{remark}
$\\$
 
Let us say that we are given bounded open sets $\mathcal{B}$ and $\mathcal{C}$ such that $\overline{\mathcal{B}} \subset \mathcal{C}$. Also, $\dot{x}(t) \in H(x(t))$ can be associated with an attractor such that  $\mathcal{B}$ is its fundamental neighborhood. Then a classical way to ensure stability of \eqref{sa_sri} is by projecting the iterate at every stage $n$, i.e., project $x_n$ onto $\overline{\mathcal{B}}$ whenever $x_n \notin \mathcal{C}$. This associated projective scheme is given by:
\begin{equation}
\label{sa_proj_sri}
\hat{x}_{n+1} = \underset{\footnotesize{\mathcal{B}, \mathcal{C}}}{\pro} \left(  \hat{x}_n + a(n) [ y_n + M_{n+1}]\right), 
\end{equation}
where $\hat{x}_0 = x_0$ and $\underset{\footnotesize{\mathcal{B}, \mathcal{C}}}{\pro}: \mathbb{R}^d \to \{\text{subsets of }\mathbb{R}^d\}$ is the projection operator that projects onto the boundary of $\mathcal{B}$, when the operand escapes from set $\mathcal{C}$, i.e.,
\[
	\underset{\footnotesize{\mathcal{B}, \mathcal{C}}}{\pro}(x) :=  \begin{cases}
	\{x\} \text{, if $x \in \mathcal{C}$}  \\
	\{y \mid d(y, x) = d(x, \overline{\mathcal{B}}), \ y \in \overline{\mathcal{B}} \}  \text{, otherwise}.
	\end{cases}.
\]
The advantage of using \eqref{sa_proj_sri} as opposed to \eqref{sa_sri} is that stability is trivially ensured since $\mathcal{B}$ is bounded. The main drawback, however, is that $\mathcal{B}$ and $\mathcal{C}$ cannot be easily constructed. Further, if $\mathcal{B}$ and $\mathcal{C}$ are not carefully chosen, then the algorithm may converge to an undesirable limit. Recall that we constructed bounded sets $\mathcal{B}_{a/b/c}$ and $\mathcal{C}_{a/b/c}$, assuming $(A3a/b/c)$ and satisfying conditions (Ca) and (Cb). 
 The tuple $(\mathcal{B}_{a/b/c}$,\ $\mathcal{C}_{a/b/c})$ can be used to obtain a hypothetical realization of \eqref{sa_proj_sri}, which is the primary tool in the stability analysis of \eqref{sa_sri}. A
 key step in this analysis involves ``comparing'' the right-hand sides of the iterations \eqref{sa_sri} and \eqref{sa_proj_sri}. To facilitate such a comparison, we make the natural assumption that the realizations of the martingale noise are identical in both the iteration and its projective counterpart. Being that \eqref{sa_proj_sri} is hypothetical, this assumption is fair.
 
We are now ready to present our final assumption which specifies the relationship between a set-valued SA and its projective counterpart.
\begin{itemize}
\item[\textbf{(A4)}]  Almost surely, there exists $N < \infty$ such that $\underset{n \ge N}{\sup} \ \lVert x_{n+1} -  \tilde{x}_{n+1} \rVert < \infty$, where $\tilde{x}_{n+1} = \left(  \hat{x}_n + a(n) [ y_n + M_{n+1}]\right)$ is the value \eqref{sa_proj_sri} at time $n+1$ before projection. The sequences $\{x_n\}$ and $\{\hat{x}_n\}$ are generated by \eqref{sa_sri}, and  \eqref{sa_proj_sri} using sets $\mathcal{B}_{a/b/c}$ and $\mathcal{C}_{a/b/c}$, respectively.
\end{itemize}

In Section~\ref{sec_ana_avi}, we show that the AVI given by \eqref{avi_avi} satisfies (A3c) and (A4). The verifiability of (A4) for general problems involving set-valued mean-fields, is discussed in Section \ref{sec_note}. Before proceeding, we define the following:


\textbf{Inward directing set:} \textit{Given a differential inclusion $\dot{x}(t) \in H(x(t))$,
an open set $\mathcal{O}$ is said to be an inward directing set with respect to the aforementioned
differential inclusion, if $\Phi_t(x) \subseteq \mathcal{O}$, $t >0$, whenever $x \in \partial \mathcal{O}$, where $\partial \mathcal{O}$ denotes the boundary of set $\mathcal{O}$.} In words, any solution starting at the boundary of $\mathcal{O}$ is ``directed inwards'', into $\mathcal{O}$.

\begin{proposition}
 \label{sa_inward}
 $\mathcal{C}_i$ is an inward directing set with respect to $\dot{x}(t) \in H(x(t))$, and a fundamental neighborhood of attractor $\mathcal{A}$, where $i \in \{a,b,c\}$.
\end{proposition}
\begin{proof} {\footnotesize PROOF:}
Fix $i \in \{a,b,c\}$. Recall that $\mathcal{C}_i = \mathcal{V}_{r(i)}$ for an appropriately chosen $r(i) > 0$. Further, recall that $\overline{\mathcal{C}}_i \subseteq \mathcal{V}_{\overline{r(i)}}$, where $ \mathcal{V}_{\overline{r(i)}}$ is a compact (closed and bounded) subset of the basin of attraction for attractor $\mathcal{A}$. Hence $\mathcal{V}_{\overline{r(i)}}$, and consequently $\mathcal{C}_i$, is a fundamental neighborhood of $\mathcal{A}$. Recall that $\mathcal{V}_{r(i)} = \{x \mid V(x) < r(i)\}$, and that $\mathcal{V}_{\overline{r(i)}} = \{x \mid V(x) \le r(i)\}$. Since $\overline{\mathcal{C}}_i \subseteq \mathcal{V}_{\overline{r(i)}}$, we have that $V(x) = r(i)$ whenever $x \in \partial \mathcal{C}_i$. In other words, the Lyapunov function $V$ evaluated at any point on the boundary of $\mathcal{C}_i$ equals $r(i)$. If follows from the assumption (A3i) that $V(x) > V(y)$ for $y \in \Phi_t(x)$, $x \in \overline{\mathcal{U}} \setminus \Lambda$ and $t > 0$. Hence, for $x \in \partial \mathcal{C}_i$ and $t >0$, we have that $V(\Phi_t(x)) < r(i)$ and $\Phi_t (x) \subseteq \mathcal{V}_{r(i)}$, i.e., $\mathcal{C}_i$ is inward directing.
\hfill $\blacksquare$
\end{proof}$\\$

To show that \eqref{sa_sri} is stable, we compare it with the projective counterpart \eqref{sa_proj_sri}. To enable a hypothetical realization of \eqref{sa_proj_sri}, we need to ensure the existence of an
\textit{inward directing set with respect to the associated DI}.
This inward directing set could be $\mathcal{C}_a$, $\mathcal{C}_b$ or $\mathcal{C}_c$, depending on whether $(A3a)$, $(A3b)$ or $(A3c)$ is verified. 
Formally, we will prove the following stability theorem:

\begin{theorem}
\label{gn_main}
 Under $(A1),\ (AV4),\ (AV5),\ (A3)$ and $(A4)$, (\ref{sa_sri}) is stable.
\end{theorem}

Before we can arrive at the above statement, we need to prove a few auxiliary lemmata (Lemmas~\ref{san_xl} -- \ref{gn_1}). However, we defer these and the proof of Theorem~\ref{gn_main} to Section~\ref{sec_gen_ana}. For now, we assume the truth of the Theorem~\ref{gn_main}, and return to the analysis of AVI.

\section{Analyzing the Approximate Value Iteration Algorithm} \label{sec_ana_avi}
Let us analyze AVI, assuming $(AV1)-(AV5)$. Before we show that AVI converges to a fixed point of the perturbed Bellman operator, we need to show that it is stable. For this, we show that $(AV1)-(AV5)$ together imply that $(A1), \ (A3)$ and $(A4)$ are satisfied, then invoke Theorem~\ref{gn_main}. Before proceeding with the analysis, let us recall the AVI recursion:
\[
 J_{n+1} = J_n + a(n) \left[ TJ_n - J_n + \epsilon_n + M_{n+1} \right],
\]
where, from (AV3), $\limsup \limits_{n \to \infty} \lVert \epsilon_n \rVert_{\nu} \le \epsilon$ a.s. For more details on the notations, the reader is referred to Section~\ref{sec_avi}. It follows from $(AV3)$, that there exists $N < \infty$, possibly sample path dependent, such that $\sup \limits_{n \ge N} \lVert \epsilon_n \rVert_{\nu} \le \epsilon$. Since we are only interested in the asymptotic behavior of the recursion, without loss of generality, we may assume that $\sup \limits_{n \ge 0} \lVert \epsilon_n \rVert_{\nu} \le \epsilon$ a.s. 
We begin our analysis with a couple of technical lemmas. First, let us define $\nu_{max} := \max\ \{\nu_i \mid 1 \le i \le d\}$ and $\nu_{min} := \min\ \{\nu_i \mid 1 \le i \le d\}$, then for $z \in \mathbb{R}^d$ we have:
$$
\frac{\lVert z \rVert_\infty}{\nu_{max}} \le \lVert z \rVert_\nu \le \frac{\lVert z \rVert_\infty}{\nu_{min}}
$$
\begin{lemma}
\label{avi_cc}
 $B^\epsilon := \{ y \mid \lVert y \rVert_\nu \le \epsilon\}$ is a convex compact subset of 
 $\mathbb{R}^d$, where $\epsilon > 0$.
\end{lemma}
\begin{proof} {\footnotesize PROOF:}
First we show that $B^\epsilon$ is convex.
 Given $y_1, y_2 \in B^\epsilon$ and $y = \lambda y_1 + (1-\lambda)y_2$, where $\lambda \in (0,1)$,
 we need show that $y \in B^\epsilon$. This is a direct consequence of: $\lVert y \rVert_\nu \le \lambda \lVert y_1 \rVert_\nu
 + (1-\lambda) \lVert y_2 \rVert_\nu \le \lambda \epsilon + (1-\lambda) \epsilon$.
 
 Now we show that $B^\epsilon$ is compact. Since $\frac{\lVert y \rVert_\infty}{\nu_{max}} \le \lVert y \rVert_\nu$, it follows that $B^\epsilon$
 is a bounded set. It is left to show that $B^\epsilon$ is closed. Let $y_n \to y$ and 
 $y_n \in B^\epsilon$ for every $n$. 
 Since $\liminf \limits_{n \to \infty} \lVert y_n \rVert_\nu \ge \lVert y \rVert_\nu$,
 it follows that $y \in B^\epsilon$. \hfill $\blacksquare$
\end{proof}$\\$
\begin{lemma}
 \label{avi_marchaud}
 The set-valued map $\tilde{T}$ given by $\tilde{T}x \mapsto Tx + B^\epsilon$ is a Marchaud map.
\end{lemma}
\begin{proof} {\footnotesize PROOF:}
 Since $B^\epsilon$ is a compact convex set, it follows that $\tilde{T}x$ is compact and convex. Since
 $T$ is a contraction map, we have:
 \[
  \lVert Tx \rVert_\nu \le \lVert T0 \rVert_\nu + \lVert x - 0 \rVert_\nu.
 \]
Given $z \in \mathbb{R}^d$, we have:
\[
 \lVert z \rVert \le  \sqrt{d} \ \lVert z \rVert_\infty \le (\sqrt{d} \ \nu_{max}) \ \lVert z \rVert_\nu,
\]
where $\lVert z \rVert = \sqrt{z_1^2 + \ldots + z_d^2}$ is the standard Euclidean norm, and $d \ge 1$.
Hence we have
\[
 \lVert Tx \rVert \le (\sqrt{d} \ \nu_{max}) \lVert Tx \rVert_\nu \le 
 (\sqrt{d} \ \nu_{max})\ \left( \lVert T0 \rVert_\nu + \lVert x \rVert_\nu \right) \text{ and that}
\]
\begin{equation}
\label{marchaud_1}
 \lVert Tx \rVert \le K'(1 + \lVert x \rVert),
\end{equation}
where $K' := \left(\frac{\sqrt{d} \ \nu_{max}}{\nu_{min}}\right) \vee \left(d \ \nu_{max} \lVert T0 \rVert_\nu \right)$.
We have that $\sup \limits_{z \in \tilde{T}x}\lVert z \rVert \le \lVert Tx \rVert + \underset{z \in B^\epsilon}{\sup}\ \lVert z \rVert$.
It follows from Lemma~\ref{avi_cc} that $K := K' \vee \underset{z \in B^\epsilon}{\sup}\ \lVert z \rVert$
is finite, hence $\sup \limits_{z \in \tilde{T}x}\lVert z \rVert \le K(1 + \lVert x \rVert)$.

We now show that $\tilde{T}$ is upper semicontinuous. Let $x_n \to x$, $y_n \to y$ and $y_n \in \tilde{T}x_n$
for $n \ge 0$. Since $T$ is continuous, we have that $Tx_n \to Tx$. Hence, $\left( y_n - Tx_n \right) 
\to \left(y - Tx \right)$
and $\epsilon \ge \liminf \limits_{n \to \infty} \lVert  y_n - Tx_n \rVert_\nu \ge \lVert y-Tx \rVert_\nu$.
In other words, $y \in \tilde{T}x$. \hfill $\blacksquare$
\end{proof}
$\\$

Let us define the Hausdorff metric with respect to the weighted max-norm $H_\nu$.

\noindent
\textbf{Definition:} \textit{Let us suppose that we are given $A, B \subseteq \mathbb{R}^d$. The Hausdorff metric with respect to the weighted max-norm $\lVert \cdotp \rVert_\nu$, is given by
$H_\nu(A,B) := \sup \limits_{x \in A} d_\nu(x,B) \vee 
\sup \limits_{y \in B} d_\nu(y,A)$, where $d_\nu(x,B) := \inf \{ \lVert x - y \rVert_\nu \mid y \in B\}$
and $d_\nu(y,A) := \inf \{ \lVert x - y \rVert_\nu \mid x \in A\}$. The Hausdorff metric can be more generally defined with respect to any metric $\rho$ as $H_\rho (A,B) := \sup \limits_{x \in A} \overline{\rho}(x,B) \vee 
\sup \limits_{y \in B} \overline{\rho}(y,A)$, where $\overline{\rho}(y,A) := \inf \{ \rho(x, y) \mid x \in A\}$ and $\overline{\rho}(x,B) := \inf \{ \rho(x, y) \mid y \in B\}$.}

Now, we state a couple of simple claims without proofs: \\
\textbf{Claim 1:} \textit{Given $x, y \in \mathbb{R}^d$,
there exist $x^* \in \tilde{T}x$ and $y^* \in \tilde{T}y$ such that 
$\lVert x^* - y^* \rVert_\nu = H_\nu(\tilde{T}x , \tilde{T}y)$.} 

The above claim directly follows from the observation that for any $x_0 \in \tilde{T}x$
and $y_0 \in \tilde{T}y$ we have
\begin{equation}
\label{avi_2epsilon}
 \lVert x_0 - y_0 \rVert_\nu \le \lVert x_0 - x^* \rVert_\nu + \lVert x^* - y^* \rVert_\nu +
 \lVert y^* - y_0 \rVert_\nu \le 2\epsilon + H_\nu(\tilde{T}x, \tilde{T}y).
\end{equation}
\textbf{Claim 2:} \textit{For any $z \in \mathbb{R}^d$ we can show that}
\begin{equation}
\label{avi_norms}
 \nu_{min} \lVert z \rVert_\nu \le \lVert z \rVert_\infty \le \lVert z \rVert \le (\sqrt{d} \ \nu_{max}) \lVert z \rVert_\nu.
\end{equation}

Given a set-valued map $H: \mathbb{R}^d \to \{\text{subsets of } \mathbb{R}^d \}$, $x \in \mathbb{R}^d$ is an equilibrium point of $H$ iff the origin belongs to $H(x)$. Since $J^*$ is the unique fixed point of the Bellman operator $T$, $J$ is a fixed point of the perturbed Bellman operator $\tilde{T}$ (defined in the statement of Lemma~\ref{avi_marchaud}), or an equilibrium point of the $\tilde{T}J - J$ operator, iff $\lVert TJ - J \rVert_\nu \le \epsilon$. In other words, \textbf{the equilibrium set of $\tilde{T}J - J$ is given by $\mathcal{A} = \{J \mid \lVert TJ - J \rVert_\nu \le \epsilon\}$}. Since $J^* \in \mathcal{A}$ and $T$ is continuous, $\mathcal{A}$ constitutes a closed neighborhood of $J^*$. Controlling the norm-bounds, $\epsilon$, on the approximation errors is one way to control the size of $\mathcal{A}$. 

Assumption $(AV2)$ dictates that $J^*$ is the global attractor (global asymptotic stable equilibrium point) of $\dot{J}(t) = TJ(t) - J(t)$. The following is a consequence of the upper-semicontinuity of $J^*$: given a neighborhood $\mathcal{N}$ of $J^*$, $\exists$ \mbox{\Large $\epsilon$}$(\mathcal{N}) > 0$ such that $\dot{J}(t) \in \tilde{T} J(t) - J(t)$ has a global attractor $\mathcal{A}' \subseteq \mathcal{N}$, provided $\epsilon \le$ \mbox{\Large $\epsilon$}$(\mathcal{N})$ a.s. \textbf{The norm-bound on the approximation errors, \mbox{\Large $\epsilon$}$(\mathcal{N})$, is therefore a function of the neighborhood $\mathcal{N}$.} We show that the asymptotic behavior of AVI, \eqref{avi_avi}, is identical to that of a solution to $\dot{J}(t) \in \tilde{T} J(t) - J(t)$, i.e., it converges to $\mathcal{A}'$. However, it is desirable for AVI to converge to an equilibrium point in $\mathcal{A}$. \textbf{In what follows, we will in fact show that (\ref{avi_avi}) converges to $\mathcal{A} \cap \mathcal{A}'$ and that 
$\mathcal{A} \cap \mathcal{A}' \neq \phi$.}

Typically, a certain degree of accuracy is expected from the AVI. This accuracy is specified through the specification of a neighborhood $\mathcal{N}$, of $J^*$. Once $\mathcal{N}$ is fixed, the above discussion provides \mbox{\Large $\epsilon$}$(\mathcal{N})$, the asymptotic norm-bound on the approximation errors, and $\mathcal{A}'$, the global attractor associated with $\dot{J}(t) \in \tilde{T} J(t) - J(t)$. Note that the asymptotic error bound is ensured by effectively training a parameterized function (for e.g., neural networks) to approximate the Bellman operator. As stated earlier, the stability analysis of AVI involves verifying that $(A1)-(A4)$ are satisfied. Among the three variants of $(A3)$, we choose to verify $(A3c)$. Recall that verifying $(A3c)$, of Section~\ref{sec_sa}, involves the construction of a global Lyapunov function, and this function can be constructed using the recipe presented in Remark~\ref{lyapunov1}. The main ingredients of this recipe are a global attractor set and an associated differential inclusion.
\textbf{For AVI, we construct the global Lyapunov function using $\mathcal{A}'$ and
$\dot{J}(t) \in \tilde{T}J(t) - J(t)$}. Note that once $(A3c)$ is verified, we construct bounded open sets $\mathcal{B}$ and $\mathcal{C}$ such that
$\mathcal{A}' \subseteq \mathcal{B}$ and $\overline{\mathcal{B}} \subseteq \mathcal{C}$. For details on this construction, the reader is referred to the discussion around \eqref{sa_proj_sri} in Section~\ref{sec_sa}. Using these sets, we can associate a projective counterpart to AVI:
\begin{align}
  \label{avi_pavi}
  \nonumber
   J_{n+1} &= \hat{J}_n + a(n)\left( T\hat{J}_n - \hat{J}_n + \hat{\epsilon}_n + M_{n+1} \right),\\
   \hat{J}_{n+1} &= \underset{\mathcal{B}, \mathcal{C}}{\pro} (J_{n+1}).
 \end{align}
 Let us call this \textbf{projective approximate value iteration}.
\textit{It is worth noting that the noise sequences in (\ref{avi_avi}) and (\ref{avi_pavi})
are identical and that $\hat{\epsilon}_n \le \epsilon$} for all $n$. 
\textbf{Following the analysis in Section~\ref{sec_analysis}, we can conclude that 
$\hat{J}_n \to \mathcal{A}'$}. We are now ready to present the main theorem of this paper.
\begin{theorem}
 \label{avi_main}
  Under $(AV1)$-$(AV5)$, the AVI recursion given by (\ref{avi_avi}) is stable and converges to an equilibrium point of $\tilde{T}J - J$, i.e., to
  some point in $\left \{J \mid \lVert TJ - J \rVert_\nu \le \epsilon \right\}$, where $\limsup \limits_{n \to \infty} \ \lVert \epsilon_n \rVert_\nu \le \epsilon$ a.s.
\end{theorem}
\begin{proof} {\footnotesize PROOF:}
To prove that AVI is stable, we begin by showing that it satisfies $(A1), \ (A3)$ and $(A4)$. Then, we invoke Theorem~\ref{gn_main} to infer the stability of AVI. Once we have stability, we proceed with the convergence analysis. \textbf{From the discussion presented before the statement of this theorem, it is clear that $(A3c)$ is satisfied}. It follows from $(AV1)$, $(AV3)$ and \eqref{avi_norms} that $(A1)$ is satisfied. Hence, we shall obtain the stability of AVI if we are able to show that $(A4)$ is satisfied. Recall that we use $J_n$ to denote the original AVI iterate, and $\hat{J}_n$ to denote its associated projective counterpart. As mentioned before, it follows from the analysis in Section~\ref{sec_analysis} that $\hat{J}_n \to \mathcal{A}'$, where $\mathcal{A}'$ is as defined in this section. Further, this analysis does not require that $(A4)$ is satisfied. 

We are now ready to show that $(A4)$ is also satisfied by AVI. Since $\hat{J}_n \to \mathcal{A}'$, there exists $N$, possibly sample path
  dependent, such that $\hat{J}_n \in \overline{\mathcal{B}}$ for all $n \ge N$. For $k \ge 0$
  and $n \ge N$,
  \[
   \lVert J_{n+k+1} - \hat{J}_{n+k+1} \rVert_\nu \le \left \lVert 
   J_{n+k} - \hat{J}_{n+k} + a(n+k) \left( (TJ_{n+k} + \epsilon_{n+k}) - 
   (T\hat{J}_{n+k} + \hat{\epsilon}_{n+k}) - (J_{n+k} - \hat{J}_{n+k}) \right)
   \right \rVert_\nu.
  \]
Grouping terms of interest in the above inequality we get:
\[
 \lVert J_{n+k+1} - \hat{J}_{n+k+1} \rVert_\nu \le (1 - a(n+k)) \lVert J_{n+k} - \hat{J}_{n+k}  \rVert_\nu
 + a(n+k) \lVert (TJ_{n+k} + \epsilon_{n+k}) - 
   (T\hat{J}_{n+k} + \hat{\epsilon}_{n+k}) \rVert_\nu.
\]
As a consequence of (\ref{avi_2epsilon}) the above equation becomes
\begin{equation}
\label{avi_thm1}
 \lVert J_{n+k+1} - \hat{J}_{n+k+1} \rVert_\nu \le (1 - a(n+k)) \lVert J_{n+k} - \hat{J}_{n+k}  \rVert_\nu
 + a(n+k) \left( 2\epsilon + H_\nu (\tilde{T}J_{n+k}, \tilde{T} \hat{J}_{n+k}) \right).
\end{equation}
We now consider the following two cases:
\\
\textit{\textbf{Case 1. $2 \epsilon \le (1-\alpha) \lVert J_{n+k} - \hat{J}_{n+k} \rVert_\nu $:}}\\
In this case (\ref{avi_thm1}) becomes
\[
 \lVert J_{n+k+1} - \hat{J}_{n+k+1} \rVert_\nu \le (1 - a(n+k)) \lVert J_{n+k} - \hat{J}_{n+k}  \rVert_\nu
 + a(n+k) \left( (1-\alpha) \lVert J_{n+k} - \hat{J}_{n+k} \rVert_\nu +
 \alpha \lVert J_{n+k} - \hat{J}_{n+k} \rVert_\nu \right).
\]
Simplifying the above equation, we get
\[
 \lVert J_{n+k+1} - \hat{J}_{n+k+1} \rVert_\nu \le \lVert J_{n+k} - \hat{J}_{n+k} \rVert_\nu.
\]
\textit{\textbf{Case 2. $2 \epsilon > (1-\alpha) \lVert J_{n+k} - \hat{J}_{n+k} \rVert_\nu $:}}\\
In this case (\ref{avi_thm1}) becomes
\[
 \lVert J_{n+k+1} - \hat{J}_{n+k+1} \rVert_\nu \le (1 - a(n+k))\frac{2 \epsilon}{1- \alpha} + a(n+k)
 \left(2 \epsilon + \alpha \frac{2 \epsilon}{1- \alpha}  \right).
\]
Simplifying the above equation, we get
\[
 \lVert J_{n+k+1} - \hat{J}_{n+k+1} \rVert_\nu \le \frac{2 \epsilon}{1- \alpha}.
\]
We may thus conclude the following:
\begin{equation} \nonumber
 \lVert J_{n+k+1} - \hat{J}_{n+k+1} \rVert_\nu \le \lVert J_{n+k} - \hat{J}_{n+k} \rVert_\nu
 \vee \left(\frac{2 \epsilon}{1- \alpha} \right).
\end{equation}
Applying the above set of arguments to $\lVert J_{n+k} - \hat{J}_{n+k} \rVert_\nu$ and proceeding
recursively to $\lVert J_{N} - \hat{J}_{N} \rVert_\nu$ we may conclude that for any $n \ge N$,
\begin{equation}
 \label{avi_thm2}
 \lVert J_{n} - \hat{J}_{n} \rVert_\nu \le \lVert J_{N} - \hat{J}_{N} \rVert_\nu
 \vee \left(\frac{2 \epsilon}{1- \alpha} \right).
\end{equation}
If we couple the above equation with \eqref{avi_norms}, we get that $(A4)$ is satisfied. It now follows from Theorem~\ref{gn_main} that $\sup \limits_{n \ge 0}\ \lVert J_{n} \rVert < \infty$ a.s. (stability of AVI).

To analyze the convergence properties of AVI, we use \textit{Theorem 3.6 and Lemma 3.8}
of Bena\"{i}m \cite{Benaim05}. Specifically, it follows from \textit{Theorem 3.6} \cite{Benaim05} that AVI converges to a closed connected internally chain transitive invariant set, $\mathcal{S}$,
of $\dot{J}(t) \in \tilde{T}J(t) - J(t)$. Further, since $\mathcal{A}'$ is a global attractor of $\dot{J}(t) \in \tilde{T}J(t) - J(t)$, it follows
that $\mathcal{S} \subseteq \mathcal{A}'$. Hence $J_n \to \mathcal{A}'$. \textbf{But we need to show that $J_n \to \mathcal{A}$, the set of equilibrium points of the set-valued operator $\tilde{T}J(t) - J(t)$}. We achieve this by showing that the $J_n$ sequence, in fact, converges to $\mathcal{A} \cap \mathcal{A}'$ and that $\mathcal{A} \cap \mathcal{A}' \neq \phi$. For this we need \textit{Theorem 2} from \textit{Chapter 6} of Aubin and Cellina \cite{Aubin} which we reproduce below.\\
\textbf{[Theorem 2, Chapter 6 \cite{Aubin}]} \textit{Let $F$ be an upper semicontinuous map from a closed subset
$\mathcal{K} \subset X$ to $X$ with compact convex values and $x(\cdotp)$ be a solution trajectory
of $\dot{x}(t) \in F(x(t))$ that converges to some $x^*$ in $\mathcal{K}$. Then $x^*$ is an equilibrium
of $F$.}

Since $\sup \limits_{n \ge 0} \ \lVert J_n \rVert < \infty$ a.s.,
there exists a large compact convex set $\mathcal{K} \subseteq \mathbb{R}^d$, possibly sample path
dependent, such that $J_n \in \mathcal{K}$ for all $n \ge 0$. Further, $\mathcal{K}$ can be chosen such that
the ``tracking solution'' of $\dot{x}(t) \in \tilde{T}J(t) - J(t)$ is also inside $\mathcal{K}$.
It now follows that the conditions of the above stated theorem are satisfied. Hence every limit point of
(\ref{avi_avi}) is an equilibrium point of $\tilde{T}J - J$. In other words, 
$J_n \to \mathcal{A} \cap \mathcal{A}'$, as required. \hfill $\blacksquare$
\end{proof}
$\\$
\begin{remark}
 \label{avi_remark}
 We can show that $\left \{J \mid \lVert TJ - J \rVert_\nu \le \epsilon \right\}
\downarrow \{J^*\}$ as $\epsilon \downarrow 0$.
In other words, as the norm-bound on the asymptotic errors of the AVI decreases, the limiting set asymptotically diminishes and converges to the point $J^*$.
 \end{remark}
 
 \changes{Given a cost-to-go vector $J$, let $J(s)$, its $s^{th}$ component, denote the \emph{cost-to-go value} associated with state $s$, where $1\le s \le d$ and $d$ is the number of states in the system. Suppose $J_\infty$ is the limit of AVI, then it follows from Theorem~\ref{avi_main}, that $\frac{ \lvert TJ_\infty (s) - J_\infty (s) \rvert}{\nu_s} \le \epsilon$. For the discounted cost problem, we know that the Bellman operator is contractive with respect to the $\ell_\infty$ norm, i.e., $\nu_s = 1$ for $1 \le s \le d$. Then, $ \lvert TJ_\infty (s) - J_\infty (s) \rvert \le \epsilon$ for all states $s$. 
 Given state $s$ and an asymptotic bound of $\epsilon$ on the approximation errors, in the weighted max-norm sense, Theorem~\ref{avi_main} states that the limit of AVI $J_\infty$, satisfies $ \lvert TJ_\infty (s) - J_\infty (s) \rvert \le \nu_s \epsilon$. Hence, lower is the weight $\nu_s$ associated with state $s$, the closer $J_\infty (s)$ is to the optimal cost-to-go value $J^*(s)$.}
 
\subsection{A note on the norm used to bound the approximation errors} \label{sec_avi1}
Recall that the approximation errors are asymptotically bounded in the weighted max-norm sense.
These errors are consequences of Bellman operator approximations, used to counter \textit{Bellman's curse of dimensionality}. In a model-free setting, typically
one is given data of the form $(x_n, v_n)$, where $v_n$ is an unbiased estimate of 
the objective function at $x_n$. The online training
of an approximation operator can be emulated by a supervised learning algorithm. This algorithm would
return a good fit $g$, of the Bellman operator, from within a class of possible functions $\mathcal{T}$. The objective for these algorithms
would be to minimize the empirical approximation errors. Previously, we considered approximation
operators that minimize errors in the weighted max-norm sense. This means, ensuring that the approximation errors are uniformly bounded across all states, something that may be hard in large-scale applications.
\textit{Here, the errors may be minimized in the weighted p-norm sense.}
\paragraph{}
In many applications the approximation operators work by minimizing the errors in the $\ell_1$ and
$\ell_2$ norms, see Munos \cite{munos} for details. Let us consider the general case
of approximation errors being bounded in the weighted p-norm sense. Specifically, we consider
(\ref{avi_avi}) with $\lVert \epsilon_n \rVert_{\omega, p} \le \epsilon$ for some fixed $\epsilon > 0$.
Recall the definition of the weighted p-norm of a
given $z \in \mathbb{R}^d$:
\[
 \lVert z \rVert_{\omega,p} = \left( \sum \limits_{i=1}^d \omega_i |z_i|^p \right)^{1/p},
\]
where $\omega = (\omega_1, \ldots, \omega_d)$ is such that $\omega_i > 0$, $1 \le i \le d$
and $p \ge 1$. 
\paragraph{}
We observe the following relation between weighted p-norm and the weighted max-norm. For $z \in \mathbb{R}^d$,
we have $\frac{\lVert z \rVert_{\infty}}{\nu_{max}} \le \lVert z \rVert_{\nu} \le \frac{\lVert z \rVert_{\infty}}{\nu_{min}}$, where $\lVert \cdotp \rVert_{\infty}$ denotes the (unweighted) max-norm, $\nu_{max} = \max \ \{\nu_1, \ldots, \nu_{d} \}$, $\nu_{min} = \min \ \{\nu_1, \ldots, \nu_{d} \}$, and $\lVert z \rVert_{\nu} = \max \left\{\frac{\lvert z_i \rvert}{\nu_i} \  \mid \  1 \le i \le d \right\}$.
Then, 
\begin{equation}
 \label{avi_pnorms}
 \nu_{min} \ \omega_{min}^{1/p} \ \lVert z \rVert_{\nu} \le \omega_{min}^{1/p} \lVert z \rVert_{\infty} \le  \lVert z \rVert_{\omega,p},
\end{equation}
where $\omega_{min} := \min \{\omega_1, \ldots, \omega_d\}$.
Let us consider the following stochastic iterative AVI scheme:
\begin{equation}
 \label{avi_avi1}
 J_{n+1} = J_n + a(n) \left(TJ_n - J_n + \tilde{\epsilon}_n + M_{n+1} \right),
\end{equation}
where $\tilde{\epsilon}_n = \overline{T}J_n - T J_n$, $\overline{T}$ is the approximate Bellman operator when the approximation errors are bounded in the weighted p-norm sense, and $\lVert \tilde{\epsilon}_n \rVert_{\omega, p} \le \epsilon$ for all $n \ge 0$. Using the previously discussed relation between the weighted p-norm and the weighted max-norm we may arrive at a similar result for \eqref{avi_avi1}. Before stating the theorem, we recall that the Bellman operator $T$ is contractive with respect to the weighted max-norm $\lVert \cdotp \rVert_{\nu}$.
\changes{
\begin{theorem}
 \label{avi_main1}
 Under (AV1,2,4,5), the AVI given by (\ref{avi_avi1}) is stable and converges to
  some point in $\left \{J \mid \lVert TJ - J \rVert_{\omega, p} \le  \epsilon \right\}$, 
 provided $\underset{n \to \infty}{\limsup} \ \lVert \tilde{\epsilon} \rVert_{\omega, p} \le \epsilon$  and $\omega_{min} > 0$.
\end{theorem}
\begin{proof}{\footnotesize PROOF:}
First, we show that \eqref{avi_avi1} is stable, i.e., $\sup \limits_{n \ge 0} \ \lVert J_n \rVert_{\omega, p} < \infty$ a.s. To do this we emulate the proof of Theorem~\ref{avi_main}, in particular, the part of the proof  that shows the stability of \eqref{avi_avi}. The main difference between \eqref{avi_avi} and \eqref{avi_avi1} is that the approximation errors in the former are bounded in the weighted max-norm sense, while in the latter they are bounded in the weighted p-norm sense. To emulate the proof, it is sufficient to show that:
\begin{equation}\label{main_limsup}
\limsup \limits_{n \to \infty } \ \lVert \tilde{\epsilon}_n \rVert_{\nu} \le \frac{\epsilon}{\nu_{min}\  \omega_{min}^{1/p}}.
\end{equation}
This is because, (AV3) is now valid with $\frac{  \epsilon}{\nu_{min}\  \omega_{min}^{1/p}}$ replacing $\epsilon$ in its statement. 
Given that $\limsup \limits_{n \to \infty } \ \lVert \tilde{\epsilon}_n \rVert_{\omega, p} \le \epsilon$, from previous discussion on norms, we get:
$$
\lVert \tilde{\epsilon}_n \rVert_{\nu} \le \frac{ \lVert \tilde{\epsilon}_n \rVert_{\omega, p}}{\nu_{min}\  \omega_{min}^{1/p}}, \forall \ n \ge 0. $$

Hence, $\sup  \limits_{n \ge 0} \ \lVert J_n \rVert_{\nu} < \infty$ a.s. Again, using the previously discussed norm relations and $\omega_{min} > 0$, we get the required $\sup  \limits_{n \ge 0}\  \lVert J_n \rVert_{\omega, p} < \infty$ a.s.

It follows from \textit{Theorem 3.6} of \cite{Benaim05} that \eqref{avi_avi1} \textit{tracks} a solution to $\dot{J}(t) \in \overline{T}J(t) - J(t)$, where $\overline{T}J = TJ + \tilde{B}_\epsilon$ and $\tilde{B}_\epsilon := \{\delta \mid \lVert \delta \rVert_{\omega, p} \le \epsilon \}$. By ``tracking'', we mean that the asymptotic properties of \eqref{avi_avi1} and $\dot{J}(t) \in \overline{T}J(t) - J(t)$ are identical. Now, as in the proof of Theorem~\ref{avi_main}, we invoke \textit{Theorem 2} from \textit{Chapter 6} of \cite{Aubin} to deduce that $J_n \to$ an equilibrium point of $\overline{T}J - J$, as $n \to \infty$. Hence, $0 \in \overline{T}J_\infty - J_\infty$, $0 \in TJ_\infty - J_\infty + \tilde{B}_\epsilon$, and $\lVert TJ_\infty - J_\infty \rVert_{\omega, p} \le \epsilon$, as required.

\hfill $\blacksquare$
\end{proof}$\\$

In \cite{munos}, a supervised learning algorithm is described, to approximate the Bellman operator. An important step in the algorithm is the sampling of states. When a state is sampled well, the approximate Bellman operator evaluated at that state has a high accuracy. Further, the weight $\omega_s$, from the above discussed weighted p-norm, is directly proportional to sampling rate of state $s$, where $1 \le s \le d$. A larger $\omega_s$ therefore corresponds to a better approximation at $s$. In the proof of Theorem~\ref{avi_main1}, we required $\omega_{min} > 0$ to establish the stability of \eqref{avi_avi1}. Let us suppose we use the approximation algorithm from \cite{munos} within our AVI routine. Then, it follows from the previous discussion that in order to guarantee the stability of AVI, one must ensure that the approximation algorithm samples every state (possibly unequally). On the other hand, when state $s$ is never sampled, $\omega_s$ and (hence) $\omega_{min}$ equal $0$. We expect the approximation errors to be high for $s$. In this case, stability of AVI cannot be guaranteed as it may be that $J_n(s) \to \infty$ with positive probability, primarily due to the unchecked accumulation of approximation errors over time.

Given a state $s$ with $\omega_s > 0$, it follows from Theorem~\ref{avi_main1} that $\left \lvert T J_\infty (s) - J_\infty (s) \right \rvert \le \nicefrac{\epsilon}{\omega_s ^{1/p}}$. Hence, larger the weight value $\omega_s$, the closer $ J_\infty (s)$ is to the optimal cost-to-go value $J^*(s)$. Suppose $s$ is associated with a very small weight value, i.e., $\omega_s \approx 0$, $\left \lvert T J_\infty (s) - J_\infty (s) \right \rvert$ might be very large. This is hardly surprising, as the approximate Bellman operator evaluated at $s$ is expected to have very high approximation errors. To summarize, if one were to use the supervised learning routine from \cite{munos}, then AVI is stable when \emph{all the states are sampled a large number of times}. With respect to the optimality of the limit, it varies over the state space. In particular, the limit of AVI, $J_\infty$, evaluated at states that are very well sampled will be close to their optimal cost-to-go values.

}
\subsection{Relevance to literature}\label{avi_munos}
In solving large-scale sequential decision making problems AVI methods are a popular choice, due to their versatility and simplicity. Further, they facilitate in finding close-to-optimal solutions despite approximation and sampling errors. Ensuring the stablitiy of such methods is a challenge in many reinforcement learning applications. 
A major contribution of this paper,
not addressed in previous literature, is the development of easily verifiable sufficient 
conditions for the almost sure boundedness of AVI methods involving set-valued dynamics.

An important contribution to the understanding of AVI methods has been due to Munos \cite{munos}. Here, the infinite horizon discounted cost problem is considered and solved using AVI. The analysis is performed when the approximation errors
are bounded in the weighted p-norm sense, a significant improvement over \cite{BertsekasBook} that only considered
max norms. In particular, a strong rate of convergence result is presented. However, the basic procedure considered is a numerical AVI scheme where complete
knowledge of the `system model', i.e., the transition probabilities, is assumed. In addition, smoothness restrictions are imposed on the transition kernel. Such requirements are hard to verify within RL settings. For the asymptotic analysis presented herein, we do not require such assumptions.

We believe that the structure of iteration \eqref{avi_avi} is generic and is observed in many deep learning algorithms. For example, it occurs in the Q-learning procedure where the Bellman operator $T$ is in fact
the Q-Bellman operator. Here,
no information is known about the system transition probabilities, i.e., we are in the `model-free' setting. Our analysis is readily applicable to this scenario. \textit{Finally, note that our analysis works for both
stochastic shortest path and infinite horizon discounted cost problems.}

To summarize, we analyze AVI-like schemes for which (a) information on the transition probabilities
is unknown, (b) there
is a measurement error (albeit asymptotically bounded) that may arise for instance from the use of function approximation, and (c) the basic framework involves either the stochastic shortest path or the discounted cost setting.

\section{Stability analysis of general set-valued stochastic approximations} \label{sec_gen_ana}
Now that we have completed the analysis of AVI, we present a stability analysis of general set-valued stochastic approximations, under the Lyapunov stability assumptions presented in Section~\ref{sec_sa}. Before we begin, we recall  \eqref{sa_sri}, the set-valued stochastic iterate: 
\begin{equation} \nonumber
x_{n+1} = x_n + a(n) \left[ y_n + M_{n+1} \right], \tag{\eqref{sa_sri} recalled}
\end{equation}
where $y_n \in H(x_n)$ for all $n \ge 0$ with $H: \mathbb{R}^d \to \{\text{subsets of }\mathbb{R}^d\}$,
$\{a(n)\}_{n \ge 0}$ is the given step-size sequence and $\{ M_{n+1}\}_{n \ge 0}$ is the given noise sequence. Let us also recall the projective counterpart of \eqref{sa_sri}, given by \eqref{sa_proj_sri} for easy reference:
\begin{equation} \nonumber
\hat{x}_{n+1} = \underset{\footnotesize{\mathcal{B}, \mathcal{C}}}{\pro} \left(  \hat{x}_n + a(n) [ y_n + M_{n+1}]\right), \tag{\eqref{sa_proj_sri} recalled}
\end{equation}
where $\hat{x}_0 = x_0$ and $\underset{\footnotesize{\mathcal{B}, \mathcal{C}}}{\pro}: \mathbb{R}^d \to \{\text{subsets of }\mathbb{R}^d\}$ is the projection operator that projects onto set $\mathcal{B}$, when the operand escapes from set $\mathcal{C}$, i.e.,
\[
	\pro _{\mathcal{B,C}}(x) :=  \begin{cases}
	\{x\} \text{, if $x \in \mathcal{C}$}  \\
	\{y \mid d(y, x) = d(x, \overline{\mathcal{B}}), \ y \in \overline{\mathcal{B}} \}  \text{, otherwise}.
	\end{cases}.
\]
\subsection{Analysis of the associated projective scheme} \label{sec_analysis}
We study the stability properties of \eqref{sa_sri} by analyzing its \emph{hypothetical} projective counterpart \eqref{sa_proj_sri}. 
Recall that the main purpose of $(A3a/b/c)$ is to ensure the existence of an inward directing set. From Proposition \ref{sa_inward}, we get that $\mathcal{C}_{a/b/c}$ is the required inward directing set. Further, from previous discussions we have that $\overline{\mathcal{B}}_{a/b/c} \subset \mathcal{C}_{a/b/c}$. Since the roles of the variants $(A3a/b/c)$, $\mathcal{B}_{a/b/c}$ and $\mathcal{C}_{a/b/c}$ are indistinguishable, with a slight abuse of notation, we generically refer to them using $(A3)$, $\mathcal{B}$ and $\mathcal{C}$, respectively.

\textbf{Before proceeding, we make a quick note on notation. Previously, the projective iterates were represented by $\hat{x}_n$ and the normal iterates by $x_n$. Since we only consider the projective iterates in this section, for the sake of aesthetics, we abuse the notation slightly and simply use $x_n$ and omit the ``hat'' notation.} Hence, the projective scheme written using the new notation is given by:
\begin{equation} \label{san_proj1}
\begin{split}
\tilde{x}_{n+1} &= x_n + a(n) \left[y_n + M_{n+1} \right], \\
x_{n+1} &= z_n, \text{ with } z_n \in \underset{\mathcal{B,C}}{\pro}(\tilde{x}_{n+1}),
\end{split}
\end{equation}
where $y_n \in H(x_n)$ and $x_0 \in \underset{\mathcal{B,C}}{\pro}(\tilde{x}_{0})$, with $\tilde{x}_0 \in \mathbb{R}^d$.
Note that the initial point $\tilde{x}_0$
is first projected before starting the projective scheme.
The above equation can be rewritten as
\begin{equation} \label{san_proj}
x_{n+1} = x_n + a(n) \left[ y_n + M_{n+1} \right] + g_n,
\end{equation}
where $g_n = z_n - \left(x_n + a(n) \left[ y_n + M_{n+1} \right] \right)$.

Integral to our analysis is the construction of a linearly interpolated trajectory that has identical asymptotic behavior as the projective counterpart \eqref{san_proj}.
We begin by dividing 
$[0, \infty)$ into diminishing intervals using  the step-size sequence.
Let $t_0:=0$ and $t_n := \sum \limits_{m=0}^{n-1} a(m)$ for $n \ge 1$. 
The linearly interpolated trajectory $X_l(t)$ is defined as follows:
\begin{equation} \label{san_li}
X_l(t) := \begin{cases}
			x_n \text{, for } t = t_n  \\
			\left(1- \frac{t-t_n}{a(n)} \right) x_n + \left(\frac{t - t_n}{a(n)} \right) \tilde{x}_{n+1}  \text{, for } t \in [t_n , t_{n+1}).
		\end{cases}
\end{equation}
The above constructed trajectory is right continuous with left-hand limits, 
\textit{i.e.,} $X_l(t) = \lim\limits_{s \downarrow t} X_l(s)$ and $ \lim\limits_{s \uparrow t} X_l(s)$ 
exist. Further the jumps occur exactly at those $t_n$'s
for which the corresponding $g_{n-1}$'s are non-zero.
We also define three piece-wise constant trajectories $X_c(\cdotp)$, $Y_c(\cdotp)$ and $G_c(\cdotp)$ as follows: $X_c(t) := x_n$, $Y_c(t) := y_n$ and $G_c(t) := \sum\limits_{m=0}^{n-1} g_m$ for $t \in [t_n, t_{n+1})$. 
The trajectories $X_c(\cdotp)$, $Y_c(\cdotp)$ and $G_c(\cdotp)$ are also right continuous with left-hand limits.
We define a linearly interpolated trajectory associated with $\{M_{n+1}\}_{n \ge 0}$ as follows:
\[
W_l(t) := \begin{cases}
			\sum\limits_{m=0}^{n-1} a(m) M_{m+1} \text{ for } t = t_n\\
			\left(1- \frac{t-t_n}{a(n)} \right) W_l(t_n) + \left(\frac{t - t_n}{a(n)} \right) W_l(t_{n+1})  \text{, for } t \in [t_n , t_{n+1}).
			\end{cases}
\]
We define a few ``left-shifted trajectories'' using the above constructed trajectories. For $t \ge 0$,
\[
X_l^n(t) :=  X_l(t_n+t), 
\]
\[
X_c^n(t) := X_c(t_n+t),
\]
\[
Y_c^n(t):= Y_c(t_n + t),
\]
\[
G_c^n(t):=G_c(t_n+t) - G_c(t_n),
\]
\[
W_l^n(t):=W_l(t_n+t) - W_l(t_n).
\]
\begin{lemma}
\label{san_xl}
$X_l^n(t) = X^n_l(0) + \int \limits_0^t Y_c^n(\tau) \,d\tau + \ W_l^n(t) + \ G_c^n(t) $ for $t \ge 0$.
\end{lemma}
\begin{proof} {\footnotesize PROOF:}
Fix $s  \in [t_m, t_{m+1})$ for some $m \ge 0$. We have the following:
\[
X_l(s) = \left(1- \frac{s-t_m}{a(m)} \right) x_m + \left(\frac{s - t_m}{a(m)} \right) \tilde{x}_{m+1} ,
\]
\[
 = \left(1- \frac{s-t_m}{a(m)} \right) x_m + \left(\frac{s - t_m}{a(m)} \right) 
\left( x_m + a(m) \left[ y_m + M_{m+1} \right] \right),
\]
\[
 = x_m + \left(s - t_m \right) \left[ y_m + M_{m+1} \right] .
\]
Let us express $X_l^{n}(t)$ in the form of the above equation. Note that 
$t_n + t \in [t_{n+k}, t_{n+k+1})$ for some $k \ge 0$. Then we have the following:
\[
X_l^n(t) = x_{n+k} + (t_n + t - t_{n+k}) \left[ y_{n+k} + M_{n+k+1} \right].
\]
Unfolding $x_{n+k}$, in the above equation till $x_n\ (i.e., X^n_l(0))$, yields:
\begin{equation}
\label{san_xl1}
X_l^n(t) = X^n_l(0) + \sum\limits_{l=n}^{n+k-1} \left( a(l) \left[y_{l} + M_{l+1} \right] + g_l \right) + (t_n + t - t_{n+k}) \left[ y_{n+k} + M_{n+k+1} \right].
\end{equation}
We make the following observations:\\ 
$G_c^n(t) = \sum\limits_{l=n}^{n+k-1} g_l$, \\
$W_{l}^n(t_{n+k} -  t_n) = \sum\limits_{l=n}^{n+k-1} a(l) M_{l+1}$, \\
$W_l^n(t) = W_{l}^n(t_{n+k} -  t_n) +  (t_n + t - t_{n+k})M_{n+k+1}$ and \\
$\int \limits_0^t Y_c^n(\tau) \,d\tau = \sum\limits_{l=n}^{n+k-1} a(l) y_l +  (t_n + t - t_{n+k}) y_{n+k}$.\\
As a consequence of the above observations, (\ref{san_xl1}) becomes:
\[
X_l^n(t) = X^n_l(0) +  \int \limits_0^t Y_c^n(\tau) \,d\tau + \ W_l^n(t) + \ G_c^n(t).
\] 
\hfill $\blacksquare$
\end{proof} $\\$
Fix $T >0$. If $\{X_l^n([0,T]) \mid n \ge 0\}$ and $\{G_c^n([0,T]) \mid n \ge 0\}$ are viewed as 
subsets of $D([0,T], \mathbb{R}^d)$ equipped with the Skorohod topology, then we may use the 
Arzela-Ascoli
theorem for $D([0,T], \mathbb{R}^d)$ 
to show that they are relatively compact, see Billingsley \cite{billingsley} for details. The Arzela-Ascoli theorem for $D([0,T], \mathbb{R}^d)$ states the following:
A set $S \subseteq D([0,T], \mathbb{R}^d)$, is relatively compact if and only if the following conditions are satisfied:
\begin{itemize}
\item[] $\underset{x(\cdotp) \in S}{\sup} \ \underset{t \in [0,T]}{\sup} \lVert x(t) \rVert < \infty$,
\item[] $\lim\limits_{\delta \to 0}$ $\underset{x(\cdotp) \in S}{\sup}$ $\underset{t_1 \le t \le t_2,\ t_2 - t_1 \le \delta}{\sup} \min \left\{ \lVert x(t) - x(t_1) \rVert, \lVert x(t_2) - x(t) \rVert \right\} = 0$,
\item[] $\lim \limits_{\delta \to 0}$ $\underset{x(\cdotp) \in S}{\sup}$ $\underset{t_1, t_2 \in [0, \delta)}{\sup} \lVert x(t_2) -  x(t_1) \rVert = 0$ and
\item[] $\lim \limits_{\delta \to 0}$ $\underset{x(\cdotp) \in S}{\sup}$ $\underset{t_1, t_2 \in [T-\delta, T)}{\sup} \lVert x(t_2) -  x(t_1) \rVert = 0$.
\end{itemize}
If $\{X_l^n([0,T]) \mid n \ge 0\}$ and $\{G_c^n([0,T]) \mid n \ge 0\}$ 
are point-wise bounded and any two of their discontinuities are separated by
at least $\Delta$, for some fixed $\Delta > 0$,
then the above four conditions will be satisfied, see \cite{billingsley} for details.
\begin{lemma}
\label{san_rc}
$\{X_l^n([0,T]) \mid n \ge 0\}$ and $\{G_c^n([0,T]) \mid n \ge 0\}$ are relatively compact in 
$D([0,T], \mathbb{R}^d)$ equipped with the Skorohod topology.
\end{lemma}
\begin{proof} {\footnotesize PROOF:}
Recall from $(A2)(i)$ that $\lVert M_{n+1} \rVert \le D$, $n \ge 0$.
Since $H$ is Marchaud, it follows that 
$\underset{x \in \overline{\mathcal{C}},\ y \in H(x)}{\sup} \lVert y \rVert \le C_1$ for some $C_1 > 0$ 
and that $\underset{n \ge 0}{\sup}\ \lVert \tilde{x}_{n+1} - x_{n} \rVert \le 
\left(\sup \limits_{m \ge 0} a(m)\right) (C_1 + D)$. 
Further, $\lVert g_n \rVert \le \lVert \tilde{x}_{n+1} - x_{n} \rVert +  d(x_n , \partial \mathcal{B}) \le C_2$ for 
some constant $C_2$ that is independent of $n$. 
In other words, we have that the sequences $\{X_l^n([0,T]) \mid n \ge 0\}$ and $\{G_c^n([0,T]) 
\mid n \ge 0\}$ are point-wise bounded. It remains to show that any two discontinuities are separated. 
Let $\tilde{d} := \underset{x \in \partial \mathcal{C}}{\min}\ d(x, \overline{B})$ and $D_1 := D\ +\ \underset{x \in \overline{\mathcal{C}}}{\sup} \underset{y \in H(x)}{\sup} \lVert y \rVert$. Clearly $\tilde{d} > 0$. Define
\[
m(n) = \max \left\{ j > 0 \ \mid \ \sum \limits_{k=n}^{n+j} a(k) < \tilde{d}/D_1 \right\}.
\]
If there is a jump at $t_n$, then $x_n \in \partial \mathcal{B}$. It follows from the definition of $m(n)$ that $x_k \in \mathcal{C}$ for $n \le k \le m(n)$. In other words, there are no discontinuities in the interval
$[t_n, t_{n + m(n)})$ and $t_{n + m(n)} - t_n \ge \frac{\tilde{d}}{2 D_1}$. 
If we fix $\Delta := \frac{\tilde{d}}{2 D_1} > 0$, then any two discontinuities are separated by at least $\Delta$.  \hfill $\blacksquare$
\end{proof}$\\$
\textit{Since $T$ is arbitrary, it follows that $\{X_l^n([0,\infty)) \mid n \ge 0\}$ and $\{G_c^n([0,\infty)) \mid n \ge 0\}$ are relatively compact in $D([0,\infty), \mathbb{R}^d)$}.
Since $\{W_l^n([0,T]) \mid n \ge 0\}$ is point-wise bounded (assumption $(A2)(i)$) and continuous, it is relatively compact in $D([0,T], \mathbb{R}^d)$. It follows from $(A2)(ii)$ that
any limit of $\{W_l^n([0,T]) \mid n \ge 0\}$, in $D([0,T], \mathbb{R}^d)$, is the constant $0$ function.
Suppose we consider sub-sequences of $\{X_l^n([0,T]) \mid n \ge 0\}$ and $\left\{X^n_l(0) +  \int \limits_0^T Y_c^n(\tau) \,d\tau + \ G_c^n(T) \mid n \ge 0 \right\}$ along
which the aforementioned noise trajectories converge, then their limits are identical. 
\paragraph{}
Consider $\{m(n)\}_{n \ge 0} \subseteq \mathbb{N}$ along which $\{G_c^{m(n)}([0,T])\}_{n \ge 0}$
and $\{X_l^{m(n)}([0,T])\}_{n \ge 0}$ converge in $D([0,T], \mathbb{R}^d)$. Further, let
$g_{m(n)-1}=0$ for all $n \ge 0$. Suppose the limit of 
$\{G_c^{m(n)}([0,T])\}_{n \ge 0}$ is the constant $0$ function, then it can be shown that the limit of
$\{X_l^{m(n)}([0,T])\}_{n \ge 0}$ is
\[
 X(t) = X(0) + \int \limits_0 ^t Y(\tau) d\tau,
\]
where $X(0) \in \overline{\mathcal{C}}$ and $Y(t) \in H(X(t))$ for $t \in [0,T]$. The proof of this is
along the lines of the proof of \textit{Theorem 2 in Chapter 5.2} of Borkar \cite{BorkarBook}. Suppose
every limit of $\{G_c^{m(n)}([0,T])\}_{n \ge 0}$ is the constant $0$ function whenever
$g_{m(n)-1}=0$ for all $n \ge 0$, then every limit of $\{X_l^{m(n)}([0,T])\}_{n \ge 0}$
is a solution to $\dot{X}(t) \in H(X(t))$. Suppose we show that the aforementioned statement is true for
every $T > 0$. Then, along with $(A4)$ the stability of (\ref{sa_sri}) is implied. 
Note that the set $K : = \{n \mid g_{n} = 0\}$ has infinite cardinality since
any two discontinuities are at least $\Delta > 0$ apart.
\begin{lemma}
\label{san_gis0}
Let $K = \{ n \mid g_n  = 0\}$. Without loss of generality let $\{X_l^n([0,T])\}_{n \in K}$ and $\{G_c^n([0,T])\}_{n \in K}$
be convergent, as $n \to \infty$, in $D([0,T], \mathbb{R}^d)$. Then $X_l^n(t) \to X(0) + \int \limits_{s = 0}^{t} Y(s) \ ds + G(t)$ for $t \in [0,T]$ and
$G(\cdotp) \equiv 0$.
\end{lemma}
\begin{proof} {\footnotesize PROOF:}
We begin by making the following observations: 
\begin{itemize}
\item[(a)] $X(0) \in \overline{\mathcal{C}}$.
\item[(b)] Any two discontinuities of $X(\cdotp)$ are at least $\Delta$ apart.
\item[(c)] $G(0) = 0$.
\item[(d)] Solutions to $\dot{x}(t) \in H(x(t))$ with starting points in $\overline{\mathcal{C}}$ will 
not hit the boundary, $\partial \mathcal{C}$, later, \textit{i.e.,} they remain in the interior of 
$\mathcal{C}$. This observation is a consequence of Proposition~\ref{sa_inward}.
\end{itemize}
It follows from the above observations that $(i)$ $X(t) \in \mathcal{C}$ for small values of $t$, $(ii)$ $\tau : = \inf \{ t  \mid t > 0, X(t^+) \neq X(t^-) \}$ and $\tau > 0$. It follows from the nature
of convergence that $\exists \ \tau_n' > \tau > \tau_n$, $n \ge 0$ such that 
\[
\tau_n ' - \tau_n \to 0,
\]
\[
\lVert X_l^n(\tau_n ') - X(\tau^+) \rVert \to 0 \text{ and}
\]
\[
\lVert X_l^n(\tau_n) - X(\tau^-) \rVert \to 0.
\]
For large values of $n$, $X_l^n(\cdotp)$ has exactly one jump (point of discontinuity) 
at a point $\hat{\tau}_n \in [\tau_n, \tau_n']$.
Let $\delta := \lVert X(\tau^+) -  X(\tau^-) \rVert > 0$, then for large values of $n$ we have
\[
\lVert X_l^n(\hat{\tau}_{n}^+) - X_l^n(\hat{\tau}_{n}^-) \rVert \ge \delta / 2.
\]
Also, $X_l^n(\hat{\tau}_{n}^-)$ is not in $\mathcal{C}$ and $ X_l^n(\hat{\tau}_{n}^+)$ is in $\partial \mathcal{B}$. Further, since $\hat{\tau}_n ^- - \tau_n \to 0$, as $n \to \infty$,
it follows that
\[
X_l^n(\hat{\tau}_{n}^-) - X(\tau ^-) \to 0.
\]
Hence, $X(\tau^-) \notin \mathcal{C}$. Similarly, we have that 
$X(\tau^+) \in \partial \mathcal{B}$. 
Observe that $X([0, \tau))$ is a solution to $\dot{x}(\cdotp) \in H(x(\cdotp))$ such that $X(0) \in \overline{\mathcal{C}}$, since $G(t) = 0$ for $t \in [0,\tau)$. Further, since $\mathcal{C}$ is \textit{inward
directing}, we have that $X(t) \in \mathcal{C}$ for $t \in [0, \tau)$.
Since $X(t) \in \mathcal{C}$ for $t < \tau$ and $X(\tau ^-) \notin \mathcal{C}$ we have $X(\tau ^-) \in \partial 
\mathcal{C}$.
\paragraph{}
Since $X(0) \in \overline{\mathcal{C}}$, we have that $V(X(0)) \le R$, for some $0 < R < \infty$.
As a consequence of our choice of $\mathcal{C}$ 
($\mathcal{C}$ is $\mathcal{C}_a$/$\mathcal{C}_b$/$\mathcal{C}_c$ within the context of $(A3a)$/$(A3b)$/$(A3c)$)
we have $V(x) = V(y)$ for any $x,y
\in \partial C$, hence we may fix $R := V(x)$ for any $x \in \partial C$. 
Fix $\tau_0 \in (0, \tau)$, it follows from Proposition~\ref{sa_inward} that
$V(X(\tau_0)) < R$. Let $t_n \uparrow \tau$ such that $t_n \in (\tau_0, \tau)$ for $n \ge 1$. Without loss of generality, $X(t_n) \to X(\tau^-)$ and $V(X(t_n)) \to V(X(\tau^-))$, as $t_n \to \tau$ (else we may choose a subsequence
of $\{t_n\}_{n \ge 0}$ along which $V(X(t_n))$ is convergent). Thus, $\exists \ N$ such that $V(X(t_n)) > V(X(\tau_0))$ for $n \ge N$. Since $X([\tau_0, t_n])$ is a solution to $\dot{x}(t) \in H(x(t))$ with starting point $X(\tau_0)$,
the aforementioned conclusion contradicts $(A3a)(iii)$/$(A3b)(iii)$/$(A3c)(iii)$.
In other words, $X(\tau^-) \in \mathcal{C}$ and $\notin \partial \mathcal{C}$. Thus we have shown that there is no jump at $\tau$, \textit{i.e.,} $X(\tau^+)  = X(\tau^-)$. \hfill $\blacksquare$
\end{proof}$\\$
Suppose $(A3a)$ / $(A3b)$ holds, then it follows from \textit{Proposition 3.25} of Bena\"{i}m, Hofbauer 
and Sorin \cite{Benaim05} that there is an attractor set $\mathcal{A} \subseteq \Lambda$ such that
$\overline{\mathcal{C}}_a$ / $\overline{\mathcal{C}}_b$ is within the basin of attraction. 
Suppose $(A3c)$ holds, then $\mathcal{A}$ is the global attractor of
$\dot{x}(t) \in H(x(t))$.
\begin{lemma}
\label{san_conv2A}
The projective stochastic approximation scheme given by (\ref{san_proj1}) converges to the 
attractor $\mathcal{A}$.
\end{lemma}
\begin{proof} {\footnotesize PROOF:}
We begin by noting that $T$ of Lemma~\ref{san_gis0} is arbitrary.
Since $G \equiv 0$, after a certain number of iterations of (\ref{san_proj1}), there are no projections, 
\textit{i.e.,} $\tilde{x}_{n} = x_{n}$ for $n \ge N$. Here $N$ could be sample path dependent.
Further, it follows from Lemma~\ref{san_gis0} that the projective scheme given by (\ref{san_proj1}) tracks a solution to $\dot{x}(t) \in H(x(t))$. In other words, the projective scheme given by (\ref{san_proj1})
converges to a limit point of the $DI$, $\dot{x}(t) \in H(x(t))$.

The iterates given by (\ref{san_proj1}) are within $\mathcal{C}$ after sometime and they track a solution to $\dot{x}(t) \in H(x(t))$. Since $\mathcal{C}$ is within the basin of attraction
of $\mathcal{A}$, the iterates converge to $\mathcal{A}$. \hfill $\blacksquare$
\end{proof}$\\$

The Lemmas proven in this section yield Theorem~\ref{sc_main}, the main result concerning the stability of \eqref{sa_sri}, stated below. Then, \textbf{it follows from \textit{Theorem 3.6 and Lemma 3.8}
of Bena\"{i}m \cite{Benaim05} that \eqref{sa_sri}
converges to a closed connected internally chain transitive invariant set
associated with $\dot{x}(t) \in H(x(t))$}.
\begin{theorem}
\label{sc_main}
Under $(A1)$-$(A4)$ and $(AV4)$, the set-valued SA given by (\ref{sa_sri}) is stable (bounded almost surely).
\end{theorem}
\begin{proof} {\footnotesize PROOF:}
Let $\{ \hat{x}_n \}_{n \ge 0}$ be the iterates generated by the projective scheme
and $\{x_n\}_{n \ge 0}$ the iterates generated by (\ref{sa_sri}). 
It follows from Lemma~\ref{san_conv2A} that $\hat{x}_n \to \mathcal{A}$ such that $\mathcal{A} \subset \mathcal{B}$. In other words
there exists $N$, possibly sample path dependent, such that $\hat{x}_n \in \mathcal{B}$ for all $n \ge N$. Without loss of generality, this is the same $N$ from $(A4)$, else one can use the maximum of the two.
Since $\mathcal{B}$ is a bounded set, we have that $\underset{n \ge N}{\sup} \lVert \hat{x}_n \rVert \le 
\underset{y \in \mathcal{B}}{\sup} \lVert y \rVert < \infty$.
It now follows from $(A4)$ that $\underset{n \ge N}{\sup} \lVert x_n \rVert < \infty$.  This directly leads to the stability of (\ref{sa_sri}). \hfill $\blacksquare$
\end{proof}$\\$
\subsection{Relaxing assumption (A2)} \label{sec_gn}
The above stated Theorem~\ref{sc_main} does not guarantee the stability of AVI, \eqref{avi_avi}. This is because, the analysis involved in proving the aforementioned theorem requires that the Martingale noise satisfy $(A2)$, instead of the weaker $(AV5)$. Since only $(AV5)$ is guaranteed for AVI, we need to prove Theorem~\ref{gn_main}, stated at the end of Section~\ref{sec_sa}. It is a generalization of Theorem~\ref{sc_main}, in that it requires the less restrictive $(AV5)$.
For the benefit of the reader we recall $(AV5)$ below:
\begin{itemize}
\item[\textbf{(AV5)}] $(M_n, \mathcal{F}_n)_{n \ge 1}$ is a square integrable Martingale difference 
sequence $\bigg( E[M_{n+1} \mid \mathcal{F}_n] = 0$ and $EM_{n+1}^2 < \infty$, $n \ge 0 \bigg)$ such that
\[
 E\left[ \lVert M_{n+1} \rVert ^2 \mid \mathcal{F}_n \right] \le K(1 + \lVert x_n \rVert ^2), \text{ where $n \ge 0$ and $K > 0$.}
\]
\end{itemize}
$\\$
In the above, $\mathcal{F}_0 := \sigma\langle x_0 \rangle$ and 
$\mathcal{F}_n := \sigma \left\langle x_0, x_1, \ldots, x_n, M_1, \ldots, M_n \right\rangle$ for $n \ge 1$. Further, without loss of generality,
we may assume that $K$ in $(AV5)$ and $(A1)$ are equal (otherwise we can use maximum of the two constants).

In what follows, we try to understand the role of $(A2)$ in the analysis presented in Section~\ref{sec_analysis}. Then, we show that it can be readily modified to require the less restrictive $(AV5)$.
In analyzing the projective scheme given by (\ref{san_proj1}), assumption $(A2)$ 
is used in Lemma~\ref{san_rc}. Specifically, $(A2)(i)$ is used to show that any two discontinuities
of $\{X^l_n([0,T])\}_{n\ge 0}$ and $\{ G_c^n([0,T]) \}_{n \ge 0}$ are separated by at least $\Delta > 0$. 
We show that the aforementioned property holds when $(A2)$
is replaced by $(AV5)$. First, we prove an auxiliary result.

\begin{lemma} \label{gn_1}
Let $\{t_{m(n)}, t_{l(n)}\}_{n \ge 0}$ be such that $t_{l(n)} > t_{m(n)}$, $t_{m(n+1)} > t_{l(n)}$ 
and $\underset{n \to \infty}{\lim} (t_{l(n)} - t_{m(n)}) = 0$. Fix an arbitrary $c > 0$ and consider the following:
\[
\psi_n := \left \lVert \sum \limits_{i = m(n)}^{l(n)-1} a(i) M_{i+1} \right \rVert.
\]
Then $P \left( \{\psi_n > c\}\  i.o. \right) = 0$ within the context of 
the projective scheme given by (\ref{san_proj1}).
\end{lemma}
\begin{proof} {\footnotesize PROOF:}
We shall show that $\sum \limits_{n \ge 0} P(\psi_n > c) < \infty$. It follows from Chebyshev's 
inequality that
\[
P(\psi_n > c) \le \frac{E \psi^2 _n}{c^2} = 
\frac{E \left[ \left \lVert \sum \limits_{i = m(n)}^{l(n)-1} a(i) M_{i+1} \right \rVert ^2 \right]}{c^2}.
\]
Since $\{M_{n+1}\}_{n \ge 0}$ is a martingale difference sequence, we get:
\begin{equation}
\label{sc_che}
P(\psi_n > c) \le \frac{ \sum \limits_{i = m(n)}^{l(n)-1} a(i)^2 E \left[\lVert M_{i+1}\rVert^2 \right]}{c^2}.
\end{equation}
Within the context of the projective scheme given by (\ref{san_proj1}), 
almost surely $\forall \ n, \ x_n \in \overline{\mathcal{C}}$, \textit{i.e.,} $\underset{n \ge 0}{\sup} \ \lVert x_n \rVert \le C_1 < \infty$ a.s.
It follows from $(AV5)$ that $E \left[ \left \lVert M_{n+1} \right \rVert ^2 \right] \le K \left(1 + E \lVert x_n \rVert ^2  \right)$. Hence,
$E \left[ \left \lVert M_{n+1} \right \rVert ^2 \right] \le K \left(1 + C_1 ^2  \right)$. 
Equation (\ref{sc_che}) becomes
\[
P(\psi_n > c) \le \frac{ \sum \limits_{i = m(n)}^{l(n)-1} a(i)^2 K \left(1 + C_1 ^2  \right)}{c^2}.
\]
Since $t_{l(n)} > t_{m(n)}$ and $t_{m(n+1)} > t_{l(n)}$, 
we have $\sum \limits_{n \ge 0} \sum \limits_{i = m(n)}^{l(n)-1} a(i)^2 \le \sum \limits_{n \ge 0} a(n)^2$.
Finally we get,
\[
\sum \limits_{n \ge 0} P(\psi_n > c) \le \frac{ \left(\sum \limits_{n \ge 0} a(n)^2 \right) 
K \left(1 + C_1 ^2  \right)}{c^2} < \infty.
\] The claim now follows from the Borel-Cantelli lemma. \hfill $\blacksquare$
\end{proof} $\\$
Let us consider the scenario in which we cannot find $\Delta$, 
the least separation between any two points of discontinuity. 
In other words, there exists $\{ t(m(n)), t(l(n)) \}_{n \ge 0}$ such that
$t_{l(n)} > t_{m(n)}$, $t_{m(n+1)} > t_{l(n)}$ and $\underset{n \to \infty}{\lim}(t_{l(n)} - t_{m(n)}) = 0$. Since we have assumed that there are no jumps between $t(m(n))$ and $t(l(n))$, we have $X_l(t_{m(n)}^+) \in \partial \mathcal{B}$ and $X_l(t_{l(n)}^-) \notin \mathcal{C}$ for all $n \ge 0$. The reader may note that every jump-point corresponds to a point in time, when the algorithm escapes $\mathcal{C}$.
We have 
\[
X_l(t_{l(n)}^-) = X_l(t_{m(n)}) + \sum \limits_{i=m(n)}^{l(n)-1} a(i) \left( y_i + M_{i+1} \right).
\]
We have that $\underset{n \ge 0}{\sup}\ \lVert y_n \rVert \le D'$ for some $0 < D' < \infty$, and $\tilde{d} = \min \limits_{x \in \partial \mathcal{C}} d(x, \overline{\mathcal{B}})$. 
The above equation becomes
\[
\lVert X_l(t_{l(n)}^-) - X_l(t_{m(n)}) \rVert \le \sum \limits_{i=m(n)}^{l(n)-1} a(i) D' +  \left\lVert \sum \limits_{i=m(n)}^{l(n)-1} a(i) M_{i+1} \right\rVert,
\]
\[
\tilde{d} \le \lVert X_l(t_{l(n)}^-) - X_l(t_{m(n)}) \rVert \le \left( t_{l(n)} - t_{m(n)} \right) D' +  \left\lVert \sum \limits_{i=m(n)}^{l(n)-1} a(i) M_{i+1} \right\rVert.
\]
Since $(t_{l(n)} - t_{m(n)}) \to 0$, for large $n$, 
$\left\lVert \sum \limits_{i=m(n)}^{l(n)-1} a(i) M_{i+1} \right\rVert > d/2$. 
This directly contradicts Lemma~\ref{gn_1}.
Hence we can always find $\Delta > 0$ separating any two points of discontinuity.
\paragraph{}
In Lemma~\ref{san_conv2A}, $(A2)$ is used to ensure the convergence of
(\ref{san_proj1}) to the attractor $\mathcal{A}$. In Theorem~\ref{sc_main}, $(A2)$ is used to
ensure the convergence of $(\ref{sa_sri})$ 
to a closed connected internally chain transitive invariant set
of the associated $DI$. Specifically, it is $(A2)(ii)$ that ensures these convergences.
Let us define $\zeta_n := \sum \limits_{k=0}^{n-1} a(k) M_{k+1}$, $n \ge 1$. 
If $\{\zeta_n\}_{n \ge 1}$ converges, then it trivially follows that the martingale noise sequence
satisfies $(A2)(ii)$. To show convergence, it is enough to show that the corresponding 
quadratic variation process converges almost surely. In other words, we need to show that 
$\sum \limits_{n \ge 0} a(n)^2 E \left( \lVert M_{n+1} \rVert ^2 | \mathcal{F}_n \right) < \infty$ a.s or
$\sum \limits_{n \ge 0} E \left( a(n)^2 \lVert M_{n+1} \rVert ^2 \right) < \infty$. Consider the
following:
\begin{equation}
 \label{gn_mart}
 \sum \limits_{n \ge 0}  a(n)^2 E\lVert M_{n+1} \rVert ^2 = 
 \sum \limits_{n \ge 0} a(n)^2 E\left[ E \left[ \lVert M_{n+1} \rVert ^2 \mid \mathcal{F}_n \right] \right]
 \le \sum \limits_{n \ge 0} a(n)^2 K(1 + E \lVert x_n \rVert ^2).
\end{equation}
Convergence of the quadratic variation process in the context of Lemma~\ref{san_conv2A}
follows from (\ref{gn_mart}) and the fact that
$E\lVert x_n \rVert ^2 \le \underset{x \in \overline{\mathcal{C}}}{\sup}\ \lVert x \rVert ^2$. In other
words,
\[
 \sum \limits_{n \ge 0} a(n)^2 E \left[ \lVert M_{n+1} \rVert ^2 \right] 
 \le \sum \limits_{n \ge 0} a(n)^2 K(1 + 
 \underset{x \in \overline{\mathcal{C}}}{\sup}\ \lVert x \rVert ^2) < \infty.
\]
Similarly, for convergence in Theorem~\ref{sc_main}, it follows from
(\ref{gn_mart}) and stability of the iterates ($\underset{n \ge 0}{\sup} \lVert x_n \rVert < \infty$ a.s.) that
\[
 \sum \limits_{n \ge 0} a(n)^2 E \left[ \lVert M_{n+1} \rVert ^2 | \mathcal{F}_n \right]
 \le \sum \limits_{n \ge 0} a(n)^2 K(1 + 
 \underset{n \ge 0}{\sup}\ \lVert x_n \rVert ^2) < \infty \ a.s.
\]
In other words, both in Lemma~\ref{san_conv2A} and Theorem~\ref{sc_main}, assumption $(A2)(ii)$
is satisfied. This gives us Theorem~\ref{gn_main}, a generalization of Theorem~\ref{sc_main}, stated at the end of Section~\ref{sec_sa}.
\subsection{Verifiability of $(A4)$} \label{sec_note}
The verifiable sufficient conditions for $(A4)$, along with the required proof are presented in the form of Lemma~\ref{note_lemma}, below. The statement of this lemma is presented for a slightly more general from of a set-valued iteration, given by $x_{n+1} \in G_n (x_n, \xi_n), \ n \ge 0$. If we define $\xi_n := M_{n+1}$, $x_n:= J_n$ and $G_n(J_n, M_{n+1}) := J_n + a(n)\left(TJ_n - J_n + \overline{B}_{\epsilon_n}(0)  + M_{n+1} \right)$, then it is easy to see that AVI, \eqref{avi_avi}, has the aforementioned set-valued iterative form.


\begin{lemma}
 \label{note_lemma}
 Let $\mathcal{B}$ and $\mathcal{C}$ be open bounded subsets of $\mathbb{R}^d$ such that $\overline{\mathcal{B}}
 \subset \mathcal{C}$. Consider the algorithm
 \[
  x_{n+1} \in G_n (x_n, \xi_n), \ n \ge 0.
 \]
We make the following assumptions:
\begin{enumerate}
 \item $\{\xi_n\}_{n \ge 0}$ is a random sequence that constitutes noise.
 \item $G_n$ is an almost-surely-bounded diameter and contractive (in the first co-ordinate with second co-ordinate fixed) set-valued map, with respect to some metric $\rho$, for $n \ge 0$. In other words, $H_\rho(G_n(x, \xi), G_n(y, \xi)) \le \alpha \rho(x,y)$ for some $0 < \alpha < 1$ and
 $\sup \limits_{u,v \in G_n(x, \xi)} \rho(u,v) \le \hat{D}$, where $0 < \hat{D} < \infty$ and $x \in \mathbb{R}^d$.
 \item The projective sequence $\{\hat{x}_n\}_{n \ge 0}$ generated by 
 \[
  \hat{x}_{n+1} \in G_n \left(\underset{\mathcal{B},\mathcal{C}}{\pro} (\hat{x}_n), \xi_n \right)
 \]
converges to some vector $x^* \in \mathcal{B}$.
\end{enumerate}
Then, almost surely, $\exists N < \infty$ such that $\sup \limits_{n \ge N} \rho(\hat{x}_n, x_n) < \infty$.
\end{lemma}
\begin{proof} {\footnotesize PROOF:}
 Since $\hat{x}_n \to x^*$ as $n \to \infty$, there exists $N$ such that $\hat{x}_n \in \mathcal{B}$
 for all $n \ge N$. For $k \ge 0$ we have the following:
 \[
  \rho(x_{n+k+1}, \hat{x}_{n+k+1}) \le 2\hat{D} + \alpha \rho(x_{n+k}, \hat{x}_{n+k}).
 \]
The arguments used to obtain the above inequality are identical to the ones used to obtain \eqref{avi_2epsilon} in the proof of Lemma~\ref{avi_cc} in Section~\ref{sec_ana_avi}.
Unfolding the right hand side down to stage $n$ we get the following:
\[
 \rho(x_{n+k+1}, \hat{x}_{n+k+1}) \le \left(1 + \alpha + \ldots + \alpha ^k \right)\ 2\hat{D} + \alpha ^{k+1} \rho(x_{n}, \hat{x}_{n}),
\]
\[
 \rho(x_{n+k+1}, \hat{x}_{n+k+1}) \le  \frac{2\hat{D}}{1-\alpha} + \alpha ^{k+1} \rho(x_{n}, \hat{x}_{n}).
\]
Since $0 < \alpha < 1$, we have
\[
 \sup \limits_{n \ge N+1} \rho(x_{n}, \hat{x}_{n}) \le \frac{2\hat{D}}{1-\alpha} + \rho(x_{N}, \hat{x}_{N}).
\]  \hfill $\blacksquare$
\end{proof}
\section{Finding fixed points of set-valued maps and Abstract Dynamic Programming} \label{sec_fp}
In Section~\ref{sec_ana_avi} we showed that the AVI algorithm given by (\ref{avi_avi}) is stable, and converges to a small neighborhood of the optimal cost-to-go vector $J^*$. For this, we started by
observing
that the fixed points of the perturbed Bellman operator belong to a small neighborhood of $J^*$
as a consequence of the upper semicontinuity of attractor sets. Then we showed that \eqref{avi_avi}
converges to a fixed point of the perturbed Bellman operator, thereby showing that 
\eqref{avi_avi} converges to a small neighborhood of $J^*$. In this section, we generalize the ideas of Sections~\ref{sec_avi} and \ref{sec_ana_avi} to develop and analyze an iterative algorithm for finding fixed points of general contractive set-valued maps.

To motivate the requirement of such a fixed point finding algorithm, we consider the field of Abstract Dynamic Programming which is an important extension of classical Dynamic Programming (DP), wherein the Bellman operator is defined in a more abstract manner, see \cite{AbstractDP}. As in classical DP, we are given: state space $\mathcal{X}$, action space $\mathcal{A}$, set of cost funtions $\mathbb{R}(\mathcal{X}) := \{J \mid J: \mathcal{X} \to \mathbb{R}\}$
and a set of valid policies $\mathcal{M} := \{\mu \mid \mu: \mathcal{X} \to \mathcal{A}\}$. Instead of a single-stage cost function which is then used to define the Bellman operator, in Abstract DP, one is given a function $H: \mathcal{X} \times \mathcal{A} \times \mathbb{R}(\mathcal{X}) \to \mathbb{R}$. The Bellman operators are defined as follows:
\[
T_\mu J(x) := H(x, \mu(x), J);
\]
\[
T J(x) := \min \limits_{\mu \in \mathcal{M}} T_\mu J(x) = \min \limits_{\mu \in \mathcal{M}} H(x, \mu(x), J).
\]
In \cite{AbstractDP}, Bertsekas has presented sufficient conditions for the existence of a fixed point of the above Bellman operator and for its optimality.  Algorithms for finding fixed points have an important role to play in Abstract DP. If we allow $H$ to be set-valued or if $H$ can only be estimated with non-zero bias, then the algorithm presented in this section can be helpful.

Suppose that we are given a set-valued map 
$T: \mathbb{R}^d \to \{\text{subsets of }\mathbb{R}^d\}$ (need not be a ``set-valued counterpart of  the Bellman operator''). We present sufficient conditions under which
the following stochastic approximation algorithm is bounded a.s. and converges to a fixed point of $T$:
\begin{equation}
 \label{fp_fp}
 x_{n+1} = x_n + a(n) \left[ y_n + M_{n+1} \right],
\end{equation}
where\\
(i) $y_n \in Tx_n - x_n$ for all $n \ge 0$. \\
(ii) $\{a(n)\}_{n \ge 0}$ is the given step-size sequence satisfying $(AV4)$.\\
(iii) $\{M_{n+1}\}_{n \ge 0}$ is the martingale difference noise sequence satisfying $(AV5)$.

$\\$
\textit{\textbf{Definitions:}} 
\begin{enumerate}
\item Given a metric space $(\mathbb{R}^d, \rho)$, recall the Hausdorff metric with respect to $\rho$
as follows:
\begin{equation}
 \nonumber
 H_\rho(A,B) := \left(\inf \limits_{x \in A} \overline{\rho}(x, B) \right) \vee \left(\inf \limits_{y \in B} \overline{\rho}(y, A)
 \right),
\end{equation}
where $A, B \subset \mathbb{R}^d$ and
$\overline{\rho}(u,C) := \min \{\rho(u,v) \mid v \in C\}$ for any $u \in \mathbb{R}^d$ and $C \subseteq \mathbb{R}^d$.
\item We call a set-valued map $T$ as contractive if and only if $H_\rho(Tx, Ty) \le \ \alpha \rho(x,y)$,
where $x, y \in \mathbb{R}^d$ and $0 < \alpha < 1$.
\item We say that $T$ is of \textit{bounded diameter}
if and only if $diam(Tx) \le \overline{D}$, $\forall \ x \in \mathbb{R}^d$ and given $0 < \overline{D} < \infty$. Here
$diam(A) := \sup \{ \rho(z_1, z_2) \mid z_1, z_2 \in A\}$ for any $A \subset \mathbb{R}^d$.
\end{enumerate}

$\\$
We impose the following restrictions on (\ref{fp_fp}):
\begin{itemize}
 \item[\textbf{(AF1)}] $T$ is a Marchaud map that is of \textit{bounded diameter} and \textit{contractive} 
 with respect to some metric $\rho$.
 \item[\textbf{(AF2)}] The metric $\rho$ is such that $\lVert x - y \rVert \le \ C\ \rho(x,y)$ for $x, y \in \mathbb{R}^d$, 
 $C>0$.
 \item[\textbf{(AF3)}] Let $F := \{x \mid x \in Tx\}$ denote the set of fixed points of $T$. There exists a
 compact subset $F' \subseteq F$ along with a strongly positive invariant bounded open neighborhood.
 \begin{center} \textit{OR} \end{center}
 $F$ is the unique global attractor of $\dot{x}(t) \in Tx(t) - x(t)$.
\end{itemize}
$\\$
Since $T$ is assumed to be contractive with respect to $\rho$, it follows from \textit{Theorem 5}
of Nadler \cite{nadler} that $T$ has at least one fixed point. Assumption $(AF2)$ is readily satisfied
by the popular metric norms such as the weighted p-norms and the weighted max-norms among others.
Assumption $(AF3)$ is imposed to ensure that (\ref{fp_fp}) satisfies $(A3b)$ or $(A3c)$. Specifically, $(AF3)$
is imposed to ensure the existence of an inward directing set associated with 
$\dot{x}(t) \in Tx(t) - x(t)$, see Proposition~\ref{sa_inward} for details. In other words, we can find 
bounded open sets $\mathcal{C}_F$ and $\mathcal{B}_F$ such that $\mathcal{C}_F$ is inward directing and
$\overline{\mathcal{B}_F} \subset \mathcal{C}_F$.
\paragraph{}
As in Section~\ref{sec_avi}, we compare (\ref{fp_fp}) with it's projective counterpart given by:
\begin{equation}
 \label{fp_proj}
 \begin{split}
  &\tilde{x}_{n+1} = \hat{x}_n + a(n) \left(y_n + M_{n+1}\right),\\
 &\hat{x}_{n+1} \in \underset{\footnotesize{\mathcal{B}_F, \mathcal{C}_F}}{\pro} (\tilde{x}_{n+1}),
 \end{split}
\end{equation}
where $y_n \in T\hat{x}_n - \hat{x}_n$, $\{ M_{n+1}\}_{n \ge 0}$ is identical for both (\ref{fp_fp})
and (\ref{fp_proj}) and $\underset{\footnotesize{\mathcal{B}_F, \mathcal{C}_F}}{\pro} (\cdotp)$ is the projection operator
defined at the beginning of Section~\ref{sec_analysis}. The analysis of the above projective
scheme proceeds in an identical manner as in Section~\ref{sec_analysis}. Specifically, we may show that
every limit point of the projective scheme (\ref{fp_proj}) belongs to $\overline{\mathcal{B}_F}$.
The following theorem is immediate.
\begin{theorem}
 \label{fp_main}
 Under $(AF1)$-$(AF3)$ and $(AV5)$, the iterates given by (\ref{fp_fp}) are bounded almost surely. Further,
 any limit point of (\ref{fp_fp}) (as $n \to \infty$) is a fixed point of the set-valued map $T$.
\end{theorem}
\begin{proof} {\footnotesize PROOF:}
 The proof of this theorem proceeds in a similar manner to that of Theorem~\ref{avi_main}.
 We only provide an outline here to avoid repetition. We begin by showing that
 (\ref{fp_fp}) is bounded almost surely (stable) by comparing it to (\ref{fp_proj}). Since the limit
 points of (\ref{fp_proj}) belong to $\overline{\mathcal{B}_F}$, there exists $N$, possibly sample
 path dependent, such that $\hat{x}_n \in \overline{\mathcal{C}_F}$ for all $n \ge N$. For $k \ge 0$,
 we have the following set of inequalities:
 \begin{equation*}
 \begin{split}
  \rho(x_{n+k+1}, \hat{x}_{n+k+1}) &\le (1 - a(n+k)) \rho(x_{n+k}, \hat{x}_{n+k}) + a(n+k)
  \left( 2D + H_\rho(Tx_{n+k}, T\hat{x}_{n+k})\right),
\\
   &\le (1 - a(n+k)) \rho(x_{n+k}, \hat{x}_{n+k}) + a(n+k)
  \left( 2D + \alpha \rho(x_{n+k}, \hat{x}_{n+k})\right),
 \end{split}
 \end{equation*}
where $diam(Tx) \le D$ for every $x \in \mathbb{R}^d$. Recall that
$0 < \alpha < 1$ is the contraction parameter of the set-valued map $T$. We consider two possible cases.
\\
\textit{\textbf{Case 1. $2D \le (1-\alpha) \rho(x_{n+k}, \hat{x}_{n+k})$ :}}
In this case, it can be shown that
\[
 \rho(x_{n+k+1}, \hat{x}_{n+k+1}) \le \rho(x_{n+k}, \hat{x}_{n+k}).
\]
\textit{\textbf{Case 2. $2D > (1-\alpha) \rho(x_{n+k}, \hat{x}_{n+k})$ :}}
In this case, it can be shown that
\[
 \rho(x_{n+k+1}, \hat{x}_{n+k+1}) \le \frac{2D}{1-\alpha}.
\]
We conclude the following:
\[
  \rho(x_n, \hat{x}_n) \le \left( \frac{2D}{1-\alpha} \right) \vee \rho(x_N, \hat{x}_N)  ,\ n \ge N.
\]
It follows from the above inequality and $(AF2)$ that (\ref{fp_fp}) satisfies assumption $(A4)$. 
Hence, we get that $\{x_n\}_{n \ge 0}$ is bounded almost surely (stable).

Since the iterates are stable, it follows from \textit{[Theorem 2, Chapter 6,\cite{Aubin}]} that
every limit point of (\ref{fp_fp}) is an equilibrium point of the set-valued map $x \mapsto Tx - x$.
In other words, if $x^*$ is a limit point of (\ref{fp_fp}), then $0 \in Tx^* -  x^*$, \textit{i.e.,}
$x^* \in Tx^*$. Hence we have shown that every limit point of (\ref{fp_fp}) is a fixed point of the set-valued
map $T$. \hfill $\blacksquare$
\end{proof}$\\$
\begin{remark}
 \label{fp_remark}
 It is assumed that $T$ is of bounded diameter, see $(AF1)$. This assumption is primarily required to show the almost sure boundedness of (\ref{fp_fp}). Specifically, it is used to show that
 $(A4)$ is satisfied. Depending on the problem at hand, one may wish to do away with this 
 ``bounded diameter'' assumption. For example, if we have $\sup \limits_{n \ge 0} \ diam(Tx_n) < \infty$
 a.s. instead, the bounded diameter assumption can be dispensed with.
 
 Since $T$ is Marchaud, it is point-wise bounded, \textit{i.e.,}
 $\sup \limits_{z \in Tx} \lVert z \rVert \le K(1 + \lVert x \rVert)$, where $K >0$. In other words,
 $diam(Tx) \le 2 K(1 + \lVert x \rVert)$. In principle, the point-wise boundedness of $T$ does allow for 
 unbounded diameters, \textit{i.e.,} $diam(Tx) \uparrow \infty$ as $\lVert x \rVert \uparrow \infty$.
 Our bounded diameter assumption prevents this scenario from happening. In applications that use
 ``approximation operators'', it is often reasonable to assume that the errors (due to approximations) are bounded. Then 
 the ``associated set-valued map'' is naturally of bounded diameter. The reader is referred to
 Section~\ref{sec_avi} for an example of this setting.
\end{remark}

\section{Conclusions}
\label{sec_conclusions}

We analyzed the stability and convergence behaviors of the Approximate Value Iteration (AVI) method from Reinforcement Learning. Such approaches utilize an approximation of the Bellman operator within a fixed point finding iteration. We modelled the approximation errors as an additive non-diminishing random process that is asymptotically bounded. We were motivated by the use of neural networks as function approximators, usually trained in an online manner to minimize these errors. Although it is improbable that the approximation errors completely vanish, it is fair to expect that they remain bounded. This is because, unbounded errors mean that the Bellman operator is approximated very poorly at some points, and that the difference between the true and the approximate operator is infinite. Evaluating the Bellman operator requires taking expectations. Within the framework of deep RL, expectations are replaced by samples. Our analysis accounts for the sampling errors by modelling them as an additive martingale difference term, which is shown to vanish asymptotically.

An important contribution of our work is providing the set of Lyapunov function based stability conditions. In addition to using them to show the stability of AVI, we presented a stability analysis of general set-valued SAs based on the aforementioned assumptions. Regarding convergence of AVI, we showed that the limit is a fixed point of the perturbed Bellman operator. Further, it belongs to a small neighborhood of the optimal cost-to-go vector $J^*$.


In the future we would like to extend the ideas in this paper to consider on-policy reinforcement learning algorithms with function approximations. Additionally, we believe that our ideas can be extended to develop and analyze algorithms that solve constrained Markov decision processes. 
Finally, it would be interesting to see if the general Lyapunov function based stability conditions can be extended to stochastic approximations with set-valued maps and Markovian noise.

\bibliographystyle{plain} 
\bibliography{references}



\end{document}